%% file: main.tex
\documentclass{article}

\usepackage[final]{cpal_2024}

\usepackage[utf8]{inputenc} %
\usepackage[T1]{fontenc}    %
\usepackage{hyperref}       %
\usepackage{url}            %
\usepackage{booktabs}       %
\usepackage{amsfonts}       %
\usepackage{nicefrac}       %
\usepackage{microtype}      %
\usepackage{xcolor}         %
\usepackage{amsmath}
\usepackage{amssymb}
\usepackage{xcolor}
\usepackage{mathtools}
\usepackage{amsthm}
\usepackage{dsfont}
\usepackage{enumitem}
\usepackage{algorithm}
\usepackage{subcaption}
\usepackage{subcaption}
\usepackage{algorithm}
\usepackage{algorithmic}

\usepackage{booktabs}
\usepackage{amsmath}
\usepackage{amssymb}
\usepackage{amsfonts}
\usepackage{amsthm}
\usepackage{dsfont} %
\usepackage{graphicx}
\usepackage{placeins}
\usepackage{soul}
\usepackage[capitalize,noabbrev]{cleveref}

\theoremstyle{plain}
\newtheorem{theorem}{Theorem}
\newtheorem{proposition}{Proposition}
\newtheorem{lemma}{Lemma}

\theoremstyle{definition}
\newtheorem{definition}{Definition}

\theoremstyle{remark}
\newtheorem{remark}{Remark}

\newcommand{\R}{\mathbb{R}}
\newcommand{\ind}{\mathds{1}}                  %
\newcommand{\Cut}{\mathrm{Cut}}                %
\newcommand{\Proj}{\operatorname{Proj}}        %
\newcommand{\onevec}{\mathbf{e}}               %
\newcommand{\front}{\mathcal{P}}               %
\newcommand{\pset}{\mathcal{Z}^{\star}}        %
\newcommand{\ndpentry}[2]{\begin{tabular}[c]{@{}r@{}}$#1$\\$\pm #2$\end{tabular}}
\title{Parallelizable Gradient-Based Optimization For Multi-Objective MaxCut}

\author{Jingjuan Huang\textsuperscript{1}, Cheng-Han Huang\textsuperscript{2}, Alvaro Velasquez\textsuperscript{3}, \\
  Rongrong Wang\textsuperscript{2}, Jia Liu\textsuperscript{1}, Ismail R. Alkhouri\textsuperscript{4,5} \\
  \textsuperscript{1}ECE Department, The Ohio State University\\
  \textsuperscript{2}CMSE Department, Michigan State University\\
  \textsuperscript{3}CS Department, University of Colorado at Boulder\\
  \textsuperscript{4}Michigan SPARC, Los Alamos National Laboratory\\
  \textsuperscript{5}MICDE, University of Michigan\\
}

\begin{document}

\maketitle

\begin{abstract}
Multi-objective combinatorial optimization arises in a wide range of problems and applications, including the canonical multi-objective MaxCut problem. Differentiable single-instance quadratic methods have recently achieved remarkable performance in single-objective combinatorial optimization. In this paper, we develop a differentiable framework for multi-objective MaxCut by combining an adjacency-based quadratic formulation with linear scalarization, thereby reducing the problem to a preference-conditioned single-objective signed-weight MaxCut problem. Theoretically, we characterize projected gradient ascent (PGA) fixed points and their local dynamics under the signed-weight adjacency formulation. We further characterize how these fixed points depend on preferences and establish their connection to Pareto optimality. Computationally, unlike conventional heuristics and branch-and-bound methods, our approach is GPU-parallelizable and can therefore benefit from substantial performance speedups. We term our algorithm Multi-objective QUadratic Combinatorial Optimization (MO-QUCO) and its parallelized variant pMO-QUCO. Empirically, we evaluate our methods on multi-layered graphs with different sizes and edge-weight distributions. Both our CPU-only and GPU-based algorithms outperform state-of-the-art exact and heuristic methods in terms of wall-clock runtime and objective quality. Despite operating under different computational settings, MO-QUCO also outperforms the SOTA quantum method.
\end{abstract}

\section{Introduction}

Combinatorial optimization (CO) considers the task of selecting the solution that optimizes a given objective from a combinatorial (often  exponentially large) set of discrete options.
Classical CO problems include the maximum cut problem (MaxCut) \citep{karp2009reducibility}, the maximum independent set problem (MIS) \citep{tarjan1977finding}, and the traveling salesman problem (TSP) \citep{dantzig1954solution}, etc.
CO finds applications in, e.g., physics and circuit layout \citep{barahona1988application}, phylogenetics \citep{snir2006using}, to name just a few.
However, most of the canonical CO problems are NP-hard \citep{karp2009reducibility}.

Moreover, in many real-world applications, CO decisions are made to optimize across different criteria, where the overall performance is evaluated by balancing multiple conflicting objectives.
Such CO problems are formulated as multi-objective CO (MO-CO) problems \citep{ehrgott2005multicriteria}.
For example, routing a fleet of delivery vehicles could trade total travel cost against workload balance across routes \citep{jozefowiez2008multi}. Since MO-CO includes single-objective CO as a special case, it is at least as hard.

In this paper, we focus on MO-MaxCut, an instance of MO-CO.
MO-MaxCut is widely used in applications, e.g., multi-criteria data clustering~\citep{du2014simultaneous} and fault-tolerant graph partitioning~\citep{censorhillel2021fault}.
To model these competing criteria, we consider multiple weighted graphs that share the same vertex and edge sets but differ in their edge weights. Our goal is to identify non-dominated cuts that represent trade-offs among the corresponding cut objectives.
The bi-criteria version of MO-MaxCut was studied classically by \citet{angel2006approximation}. Recently, \citet{kotil2025quantum} approached MO-MaxCut with quantum sampling based on Quantum Approximate Optimization Algorithm (QAOA) \citep{farhi2014quantum} and released a benchmark with hardware-native instances. In contrast, we take a differentiable approach and adopt their benchmark for evaluation.

Two challenges arise in solving the MO-MaxCut problem. \textbf{First}, since each objective is typically NP-Hard and the number of cuts grows exponentially with the number of nodes, classical algorithms are either heuristics without optimality guarantees or approximation approaches only with sub-optimality gap guarantees \citep{goemans1995improved,alkhouri2025scalable}. \textbf{Second}, the multi-objective nature further complicates the problem in that the number of trade-off relations grows exponentially with the number of objectives \citep{ehrgott2005multicriteria,allmendinger2022if}.

Although MO-CO problems have been studied from multiple angles, scalable solvers have emerged only recently through learning-based approaches (e.g., routing and knapsack problems~\citep{lin2022pareto}, or treating the objectives as black boxes \citep{singh2026divide}.)
Results for MO-MaxCut are even more limited.
Existing approaches evaluated for MO-MaxCut can be broadly categorized by computational paradigm into classical and quantum methods. Among classical approaches, exact multi-objective integer programming algorithms such as the Defining Point Algorithm (DPA-a) \citep{dachert2024simple} and the Disjunctive Constraint Method (DCM) \citep{boland2017new} iteratively solve integer programs to enumerate the non-dominated set.
However, a complete enumeration can be expensive or even computationally intractable.
Classical scalarization approaches commonly include the weighted-sum method (WSM) and the $\varepsilon$-constraint method ($\varepsilon$-CM), both of which reduce a multi-objective problem to a sequence of single-objective optimization problems \citep{ehrgott2005multicriteria}. Exact WSM is computationally convenient but generally recovers only supported Pareto-optimal solutions (cf. Definition~\ref{def:supported-outcomes} and \citep{konen2025supportedness}).
In contrast, $\varepsilon$-CM can recover unsupported solutions at the cost of solving constrained subproblems.
Scalable solvers for MO-MaxCut remain limited in the literature, with existing approaches restricted to bi-criteria approximation algorithms \citep{angel2006approximation}, a quantum algorithm \citep{kotil2025quantum}, which uses low-depth QAOA circuits to sample diverse candidate cuts and construct an approximate Pareto front through non-domination filtering.

Motivated by advances of recent differentiable solvers for solving single-objective CO problems with remarkable scalability in multiple CO problems~\citep{schuetz2022combinatorial,alkhouri2024differentiable,alkhouri2025scalable,sun2026mutation}, in this paper, we propose \ul{m}ulti-\ul{o}bjective \ul{qu}adratic box-\ul{c}onstrained \ul{o}ptimization (MO-QUCO) for the MO-MaxCut problem. Our contributions are summarized as follows:

\begin{itemize}[leftmargin=*,topsep=0pt,partopsep=0pt]
\item We propose a differentiable scalarized formulation for MO-MaxCut that enables gradient-based optimization over a continuous relaxation. The formulation combines the recent adjacency formulation of MaxCut with linear scalarization for multi-objective optimization.

\item Theoretically, we characterize the fixed points of projected gradient ascent (PGA) under signed edge weights. For a fixed preference vector, we show that the binary fixed points of the adjacency-based PGA update are precisely the cuts for which no $1$-bit flip improves the scalarized objective. Furthermore, we analyze the local dynamics near fixed points of PGA. Building on these results, we characterize the set of preference vectors for which a given cut is a fixed point and establish the corresponding connection to Pareto properties.

\item We develop a batched and GPU-parallelizable algorithm that simultaneously optimizes over a large collection of preference vectors and random initializations. The resulting non-dominated cuts are accumulated in a Pareto archive, with an optional multi-objective 1-flip refinement step.

\item By evaluating across a wide range of objective counts and graph sizes, we show that MO-QUCO and pMO-QUCO outperform state-of-the-art exact, heuristic, and quantum methods. We further conduct comprehensive ablation studies to quantify the contribution of each key algorithmic component.

\end{itemize}

\section{Preliminaries}

\paragraph{Notations.}
We denote an undirected weighted graph by $G = (V, E, W)$, where $V$ is the vertex set, $E \subseteq V \times V$ is the edge set, and $W$ is the matrix of edge weights. The number of nodes (resp.~edges) is denoted by $|V| = n$ (resp.~$|E| = m$), where $|\cdot|$ denotes the cardinality of a set unless otherwise stated.
Since $G$ is undirected, we identify each edge with an unordered pair $\{u,v\}$ with $u \neq v$. The all-ones vector of size $n$ is denoted by $\onevec_n$ and $\operatorname{tr}(\cdot)$ denotes the trace of a matrix.
The weight matrix $W \in \R^{n \times n}$ is symmetric, where $W_{uv}$ denotes the weight of edge $\{u,v\} \in E$, which may be negative, and $W_{uv} = 0$ whenever $\{u,v\} \notin E$. The corresponding degree matrix is $D_W = \operatorname{diag}(W \onevec_n)$, and the signed graph Laplacian is $L_W = D_W - W$. %
For any positive integer $K$, $[K] := \{1,\dots,K\}$. The signum function, applied element-wise, is denoted as $\operatorname{sign}(\cdot)$. We use $\ind\{\cdot\}$ to denote the indicator function, which returns $1$ (resp.~$0$) when its argument is True (resp.~False). For a box $\mathcal{B} \subset \R^n$, $\Proj_{\mathcal{B}}(\cdot)$ denotes the Euclidean projection onto $\mathcal{B}$, applied entry-wise. Finally, we use $\Delta^K := \{\lambda \in [0,1]^{K} : \onevec_K^{\top} \lambda = 1\}$ to denote the probability simplex over $K$ dimensions.

\paragraph{Problem Statement.}
A \emph{cut} of $V$ is a partition $(S, \bar{S})$ with $\bar{S} = V \setminus S$, which we encode by a binary vector $z \in \{0,1\}^{n}$ with $z_v = \ind\{v \in S\}$. An edge $\{u,v\}$ is said to be \emph{cut} when its endpoints lie on opposite sides, i.e., when $z_u \neq z_v$. The \emph{cut value} of $z$ in $G$ is
\begin{equation}
    \Cut_G(z) := \sum_{\{u,v\} \in E} W_{uv}\, \ind\{z_u \neq z_v\}.
    \label{eq:cutval}
\end{equation}

\begin{definition}[Signed-weight SO-MaxCut]
\label{def:so}
Given an undirected weighted graph $G = (V, E, W)$ whose edge weights may be positive or negative, the signed-weight single-objective (SO) MaxCut problem is $\max_{z \in \{0,1\}^{n}} \Cut_G(z).$
\end{definition}

Since $z_u, z_v \in \{0,1\}$, we have $\ind\{z_u \neq z_v\} = (z_u - z_v)^2$, which yields the quadratic representation
\begin{equation}
    \Cut_G(z)
    = \sum_{\{u,v\} \in E} W_{uv}\, (z_u - z_v)^2
    = z^{\top} L_W\, z.
    \label{eq:quadform}
\end{equation}
Equivalently, a cut $z$ corresponds to the spin vector $s = 2z - \onevec_n \in \{-1,1\}^{n}$, where each edge contributes $(s_u - s_v)^2 = 4$ if it is cut and $0$ otherwise, so the cut value can also be written as $\Cut_G(s) = \tfrac{1}{4}\, s^{\top} L_W\, s$.

In this work, we generalize the recent differentiable approach of \citet{alkhouri2025scalable}, which optimizes quadratic representations (e.g., \eqref{eq:quadform} or other objectives from \citealp{sun2026mutation}) over a relaxed box, to multi-objective optimization.
In the multi-objective (MO) setting, we consider $K \geq 2$ weighted graphs $G^{(k)} = (V, E, W^{(k)})$, $k \in [K]$, sharing the same vertex set $V$ and edge set $E$, with Laplacians
$L^{(k)} := D_{W^{(k)}} - W^{(k)}.$
Every cut is scored by the vector-valued objective
\begin{equation}
\begin{aligned}
F(z)
&:= \bigl(f_1(z), \dots, f_K(z)\bigr) \in \mathbb{R}^{K},
\qquad \text{with} \\
f_k(z)
&:= \Cut_{G^{(k)}}(z), \qquad k=1,\dots,K.
\end{aligned}
\label{eq:vecobj}
\end{equation}
In general, no single cut maximizes all $f_k$ simultaneously. Solutions are instead evaluated in the sense of Pareto dominance, which we formally define as follows:

\begin{definition}[Pareto Dominance and Pareto Front]
\label{def:pareto}
Let $z, z' \in \{0,1\}^{n}$. We say $z$ \emph{dominates} $z'$ if $f_k(z) \geq f_k(z')$ for all $k \in [K]$ and $F(z) \neq F(z')$. A cut $z^{\star}$ is \emph{Pareto-optimal} if no $z \in \{0,1\}^{n}$ dominates it.
The set of all Pareto-optimal cuts is denoted by $\pset$, and the \emph{Pareto front} is its image $\front := \{F(z) : z \in \pset\} \subset \R^{K}$.
\end{definition}

\begin{definition}[MO-MaxCut]
\label{def:mo}
Given weighted graphs $\{G^{(k)} = (V, E, W^{(k)})\}_{k \in [K]}$, the MO-MaxCut problem is to compute a set of Pareto-optimal cuts $\mathcal{Z}$ whose image covers the entire Pareto front of \eqref{eq:vecobj}, i.e., $\{F(z) : z \in \mathcal{Z}\} = \front$.
\end{definition}

Computing $\front$ of MO-MaxCut exactly is NP-hard even in the special case with $K=1$ (i.e., SO-MaxCut). Therefore, MO solvers in practice typically retain a finite archive $\hat{\mathcal{Z}} = \{z^{(1)}, \ldots, z^{(M)}\}$ of mutually non-dominated cuts, where $z^{(i)}\in\{0,1\}^{n}$ for every $i\in[M]$. We denote its image in the objective space by
\[
F[\hat{\mathcal{Z}}]
:=
\{F(z) : z \in \hat{\mathcal{Z}}\}
\subseteq \R^{K}.
\]
The archive approximates $\front$ through this image, and we quantify its approximation quality in the experiments.

\section{A Differentiable Approach For MO-MaxCut}

In this section, we start by describing the linear scalarization used to transform the MO-MaxCut problem into a signed-weight SO-MaxCut problem. Then, we characterize the PGA-fixed points of the resulting optimization. Preference-conditioned points are characterized, followed by a detailed procedure of our CPU-only and GPU-batched algorithms.
Due to space limitation, all proofs are relegated to the supplementary material.

\subsection{Scalarization \& The choice of Quadratic Formulation}

Linear scalarization is a classical technique for converting a multi-objective problem to a single-objective problem \citep{zadeh1963optimality,lin2019pareto,navon2020learning}. Given a preference vector $\lambda \in \Delta^K$, we replace the $K$ objectives in \eqref{eq:vecobj} by a weighted sum
\begin{equation}
    \label{eqn: weighted sum}
  h_\lambda(z) := \sum_{k \in [K]} \lambda_k f_k(z).
\end{equation}
In our setting, the scalarized problem is a new signed-weight SO-MaxCut problem. Specifically, linear scalarization yields the \emph{aggregated graph} $G_\lambda = (V, E, W_\lambda)$, with the weight matrix and Laplacian being:
\[W_\lambda := \sum_{k} \lambda_k W^{(k)},~~~ L_\lambda := \sum_{k} \lambda_k L^{(k)} = D_{W_\lambda} - W_\lambda,\]
since $h_\lambda(z) = \Cut_{G_\lambda}(z)$ for every cut $z$. Note that $W_\lambda$ inherits negative entries from the layers, so the aggregated problem remains a signed-weight MaxCut problem.

Following the same box-constrained relaxation of \citet{alkhouri2025scalable}, we relax vector $z$ to $x \in [-1,1]^{n}$ and maximize \[g^{L}_{\lambda}(x) := \tfrac{1}{4}\, x^{\top} L_\lambda\, x,\] which is the standard Laplacian formulation of MaxCut. Recently, for unweighted graphs, the work in \citet{sun2026mutation} introduced three quadratic MaxCut objectives: the perturbed Laplacian, the adjacency formulation, and the perturbed biased formulation. Using our aggregated graph, these correspond to
\[g^{P}_{\lambda}(x) := \tfrac{1}{4}\, x^{\top} (L_\lambda + \delta I)\, x,~~~ g^{A}_{\lambda}(x) := -\tfrac{1}{4}\, x^{\top} W_\lambda\, x,\]
\[g^{B}_{\lambda}(x) := -\tfrac{\delta}{4}\, \onevec_n^{\top} x - \tfrac{1}{4}\, x^{\top} W_\lambda\, x,\]
respectively, where $\delta > 0$ is a small perturbation parameter.

\begin{lemma}[Binary Equivalence of Quadratic Formulations]
\label{lem:binary-equivalence}
For a fixed preference vector $\lambda$ and $\delta>0$, the objectives $g^L_\lambda$, $g^P_\lambda$, and $g^A_\lambda$ differ only by additive constants on $\{-1,1\}^n$ and have the same binary maximizers. In contrast, $g^B_\lambda$ contains an additional nonconstant linear term and may have different binary maximizers.
\end{lemma}

For unweighted MaxCut, \citet{sun2026mutation} analyzed PGA fixed points and local maximizers of four box-relaxed quadratic formulations.
Motivated by Lemma~\ref{lem:binary-equivalence}, we focus on the adjacency formulation because it preserves the binary maximizers of the Laplacian formulation.
Our comparison of the four formulations also shows a good balance between solution quality and runtime for the adjacency formulation (Appendix~\ref{app:quadratic-formulations}).
In the next subsection, we extend the fixed-point analysis to the signed-weight adjacency formulation for a fixed preference vector $\lambda$.
We establish the equivalence between binary PGA fixed points and cuts with no improving $1$-bit flip (Theorem~\ref{thm:binaryfixed}), and analyze the local dynamics near PGA fixed points (Theorem~\ref{thm:adjacency}).

\subsection{PGA Fixed Points For Signed-Weight SO-MaxCut}

Throughout this subsection, we fix a preference vector $\lambda \in \Delta^K$ and consider the adjacency formulation $g^{A}_{\lambda}$.
The aggregated weight matrix $W_\lambda$ is symmetric with zero diagonal.
Given the scalarized objective $g^{A}_{\lambda}$, the PGA iterate with step size $\alpha > 0$ and gradient $\nabla g^{A}_{\lambda}(x) = -\tfrac{1}{2} W_\lambda x$ is
\begin{equation}
    x^{t+1} = \Proj_{[-1,1]^{n}}\!\bigl(x^{t} - \alpha\, W_\lambda x^{t}\bigr),
    \label{eq:pga}
\end{equation}
where the constant $\tfrac{1}{2}$ is absorbed into $\alpha$.

Next, we analyze the PGA fixed points of \eqref{eq:pga} in the signed multi-layered setting, extending the single-graph unweighted analysis of \citet{sun2026mutation}.

Under the adjacency formulation, $-W_\lambda$ may not be positive semidefinite and $g^{A}_{\lambda}$ may not be convex as there could be some $x \in \R^{n}$ such that the positive semi-definite (PSD) condition
$
-x^{\top}W_\lambda x\geq0
$
is not always satisfied. This is because a nonzero symmetric zero-diagonal $W_\lambda$ has both positive and negative eigenvalues, so $-W_\lambda$ need not be PSD.

\begin{definition}[Stationary point]
\label{def:stationary}
A point $x \in [-1,1]^{n}$ is a \emph{stationary point} of $g^{A}_{\lambda}$ if $W_\lambda x = 0$ (equivalently, $\nabla g^{A}_{\lambda}(x) = 0$), i.e., if $x$ lies in the null space $N(W_\lambda) := \{y \in \R^{n} : W_\lambda y = 0\}$ of $W_\lambda$.
\end{definition}
Every stationary point has $g^{A}_{\lambda}(x) = 0$ and is left unchanged by \eqref{eq:pga}. Since $W_\lambda 0 = 0$, the zero vector is always a stationary point. We include these stationary points in the definition of PGA fixed points.

\begin{definition}[PGA Fixed Point]
\label{def:fixed}
\label{def:adjacency-fixed}
A point $x \in [-1,1]^{n}$ is a \emph{PGA fixed point} of \eqref{eq:pga} if one update leaves it unchanged, i.e., $x = \Proj_{[-1,1]^{n}}\!\bigl(x - \alpha\, W_\lambda x\bigr)$.
The fixed point is binary if $x\in\{-1,1\}^n$. Otherwise, it is nonbinary and at least one coordinate lies in $(-1,1)$.
\end{definition}
The projection creates nonstationary PGA fixed points at the boundary of the box, where the gradient needs not vanish.

\begin{lemma}[PGA Fixed-point Condition]
\label{lem:fixedpoint}
For every $\alpha > 0$, a point $x \in [-1,1]^{n}$ is a PGA fixed point of \eqref{eq:pga} if and only if $x_v = -\operatorname{sign}\!\bigl((W_\lambda x)_v\bigr)$, $\forall v \in [n]$ with $(W_\lambda x)_v \neq 0$.
\end{lemma}

\begin{remark}
\label{rem:boundary}
Lemma~\ref{lem:fixedpoint} implies that every nonstationary PGA fixed point of \eqref{eq:pga} lies on the boundary of $[-1,1]^n$, i.e., at least one coordinate equals $-1$ or $1$. Instead, any point in the interior that is left unchanged by the PGA update must satisfy $W_\lambda x=0$ and therefore it is a stationary point.
\end{remark}

\begin{definition}[$1$-Bit Flip]
\label{def:scalarized-flip}
For $s\in\{-1,1\}^n$ and $v\in[n]$, let $s^{(v)}$ be the vector obtained by replacing $s_v$ with $-s_v$ while leaving all other coordinates unchanged. The vector $s^{(v)}$ is the $1$-bit flip of $s$ at vertex $v$. Its scalarized gain is $\Delta_v(s):=\Cut_{G_\lambda}(s^{(v)})-\Cut_{G_\lambda}(s)$. The flip is improving if $\Delta_v(s)>0$.
\end{definition}

The following theorem characterizes the binary PGA fixed points under the adjacency formulation for arbitrary signed weights.

\begin{theorem}[Characterization of Binary PGA Fixed Points]
\label{thm:binaryfixed}
For a binary vector $s\in\{-1,1\}^{n}$ and each vertex $v\in[n]$, define
$
C_v(s):=\sum_{u:\,s_u\neq s_v}(W_\lambda)_{uv},
$
which is the total signed weight of the cut edges incident to $v$. Let $d_v:=\sum_u(W_\lambda)_{uv}$ be the total signed weight of all edges incident to $v$. The vector $s$ is a PGA fixed point of \eqref{eq:pga} if and only if $C_v(s)\geq d_v/2$ for every $v\in[n]$. Since $\Delta_v(s)=s_v(W_\lambda s)_v=d_v-2C_v(s)$, this condition is equivalent to $\Delta_v(s)\leq0$ for every $v\in[n]$. Equivalently, its cut admits no improving $1$-bit flip.
\end{theorem}

Theorem~\ref{thm:binaryfixed} shows that a binary cut is a PGA fixed point exactly when the sum of the crossing-edge weights is at least half the total incident edge weight at every vertex.
If $C_v(s)<d_v/2$ for some vertex, the corresponding PGA update changes that coordinate, so $s$ is not a PGA fixed point.

The following result analyzes PGA iterates near binary and nonbinary fixed points.
\begin{theorem}[Local Dynamics near PGA Fixed Points]
\label{thm:adjacency}
\textnormal{(i)} If $\Delta_v(s)<0$ for every $v$, the iterates of \eqref{eq:pga} reach $s$ in finitely many steps from every sufficiently close $x^0\in[-1,1]^n$.

\textnormal{(ii)} For a nonbinary fixed point $x^*$, set $\mathcal I:=\{v:|x_v^*|<1\}\neq\varnothing$ and $H:=(W_\lambda)_{\mathcal I\mathcal I}$. Assume $x_v^*(W_\lambda x^*)_v<0$ for every $v\notin\mathcal I$. If $H\neq0$, iterates starting arbitrarily close to $x^*$ can leave a fixed neighborhood of $x^*$. If $H=0$, there are distinct fixed points arbitrarily close to $x^*$.
\end{theorem}

Theorem~\ref{thm:adjacency} describes the local dynamics near PGA fixed points under the adjacency formulation.
When every $1$-bit-flip gain is strictly negative, PGA iterates reach the binary fixed point in finitely many steps from all sufficiently close initial points.
In contrast, nonbinary fixed points satisfying the conditions in part~(ii) are not locally asymptotically stable.
The convergence bound and full proof are given in Appendix~\ref{app:adjacency-dynamics}.

\subsection{Preference-conditioned Fixed Points}

In the previous subsection, we fixed a preference vector and characterized the binary PGA fixed points. In this subsection, we reverse the viewpoint by fixing a cut $s\in\{-1,1\}^{n}$ and identify which preferences make the cut $s$ a binary fixed point of PGA.
Toward this end, we define the strictly positive preference simplex as
\[
\Delta_{+}^{K}
:=
\bigl\{\lambda\in\Delta^K : \lambda_k>0 \text{ for all } k\in[K]\bigr\}.
\]
\begin{definition}[Supported Cuts and Outcomes]
\label{def:supported-outcomes}
A cut $z\in\{0,1\}^{n}$ is \emph{supported} if it globally maximizes the scalarized objective $h_\lambda$ in \eqref{eqn: weighted sum} for some $\lambda\in\Delta_{+}^{K}$ .
The corresponding vector-valued objective $F(z)$ is called a \emph{supported outcome}.
\end{definition}

\begin{lemma}[All Supported Cuts Being Pareto-optimal]
\label{lem:supported-pareto}
Every supported cut $z\in\{0,1\}^{n}$ is Pareto-optimal.
\end{lemma}

We then call a Pareto-optimal cut that is not supported \emph{unsupported} and its objective vector an \emph{unsupported outcome}.

Exact WSM over $\Delta_{+}^{K}$ returns only supported outcomes because it globally maximizes $h_\lambda$. Therefore, it may not recover every Pareto-optimal cut in $\pset$ when $\front$ contains unsupported outcomes \citep{zhang2007moea}. For example, consider three Pareto-optimal outcomes $F(z^{(1)})=(3,0)$, $F(z^{(2)})=(1,1)$, and $F(z^{(3)})=(0,3)$. For any $\lambda=(\lambda_1,\lambda_2)\in\Delta_{+}^{2}$, we have $h_\lambda(z^{(2)})=1$. However, $\max\{h_\lambda(z^{(1)}),h_\lambda(z^{(3)})\}\geq 3/2$. Thus, $z^{(2)}$ is unsupported and cannot be returned by exact WSM.

Unlike exact WSM, our method does not solve each scalarized problem to global optimality. We next formalize how this approximate optimization can generate candidates beyond supported cuts.

We first characterize the preferences that make $s$ a PGA binary fixed point. Let $z=(s+\onevec_n)/2$ denote its corresponding $\{0,1\}^n$ cut encoding. We write $F(s):=F(z)$ and $h_\lambda(s):=h_\lambda(z)$. Recall from Definition~\ref{def:scalarized-flip} that $s^{(v)}$ denotes the spin vector obtained by flipping coordinate $v$. We collect the changes caused by this flip across the $K$ objectives in the \emph{flip-gain vector} defined as:
\begin{align}
    \gamma_v(s) &:= \bigl(\Delta^{(1)}_v(s), \dots, \Delta^{(K)}_v(s)\bigr) \in \R^{K}.
    \label{eq:flipgain}
\end{align}
Here, $\Delta^{(k)}_v(s):=s_v\,\bigl(W^{(k)}s\bigr)_v$ is the change in $\Cut_{G^{(k)}}$ caused by flipping vertex $v$. By linearity in the edge weights, the scalarized gain of the same flip satisfies:
\begin{equation}
h_\lambda\bigl(s^{(v)}\bigr)-h_\lambda(s)
=
\sum_{k=1}^{K}\lambda_k\Delta_v^{(k)}(s)
=
\lambda^\top\gamma_v(s).
\label{eq:scalarized-flip-gain}
\end{equation}
By Theorem~\ref{thm:binaryfixed}, the inequalities $\lambda^\top\gamma_v(s)\leq0$ for every $v\in[n]$ are exactly the conditions under which $s$ is a binary fixed point of the PGA update under the adjacency formulation. This motivates the following definition.
\begin{definition}[Adjacency Formulation PGA Fixed Point Preference Polytope]
\label{def:prefpolytope}
For a cut $s \in \{-1,1\}^{n}$, define its adjacency formulation fixed-point preference polytope as
\begin{equation}
    \Lambda^{A}(s) := \bigl\{\lambda \in \Delta^K : \lambda^{\top} \gamma_v(s) \leq 0 \ \text{for all } v \in [n]\bigr\}.
    \label{eq:prefpolytope}
\end{equation}
\end{definition}
As the intersection of the simplex with $n$ half-spaces, $\Lambda^{A}(s)$ is a convex polytope, which may be empty or lower-dimensional.
With $\Lambda^{A}(s)$, we are now ready to state the following key result:

\begin{theorem}[Preference-conditioned PGA Fixed Points]
\label{thm:prefpolytope}
Let $s \in \{-1,1\}^{n}$, $\lambda \in \Delta^K$, and $\alpha>0$. Then, the following statements are equivalent:
\[
\begin{aligned}
\textnormal{(i)}\quad
&s=\Proj_{[-1,1]^n}\!\left(s-\alpha W_\lambda s\right),\\
\textnormal{(ii)}\quad
&\lambda^\top\gamma_v(s)\leq0\ \ \forall v\in[n],
\qquad
\textnormal{(iii)}\quad
\lambda\in\Lambda^{A}(s).
\end{aligned}
\]
Consequently, $\Lambda^{A}(s)$ is precisely the set of preferences under which $s$ is a PGA fixed point of $g^{A}_{\lambda}$.
\end{theorem}
Theorem~\ref{thm:prefpolytope} provides three equivalent conditions that characterize the relationship between a binary fixed point and a preference vector. Furthermore, if $\lambda^\top\gamma_v(s)<0$ for every $v\in[n]$, PGA iterates reach $s$ in finitely many steps from sufficiently close initial points by Theorem~\ref{thm:adjacency}(i).
Next, we characterize the relationship among three classes of cuts. We define their corresponding sets as follows:
\begin{align*}
\mathcal{Z}_{\mathrm{sup}}
&:= \bigl\{s \in \{-1,1\}^{n} : \exists \lambda \in \Delta_{+}^{K}
    \text{ with}\\[-1mm]
&\qquad
h_\lambda(s)
    = \max_{t \in \{-1,1\}^{n}} h_\lambda(t)\bigr\},\\
\mathcal{Z}_{\mathrm{fix}}
&:= \bigl\{s \in \{-1,1\}^{n} :~ \Lambda^{A}(s)
    \cap \Delta_{+}^{K}
    \neq \emptyset\bigr\},\\
\mathcal{Z}_{\mathrm{loc}}
&:= \bigl\{s \in \{-1,1\}^{n} :
    \nexists v \in [n] \text{ with}\\[-1mm]
&\qquad F(s^{(v)}) \geq F(s)
    \text{ and } F(s^{(v)}) \neq F(s)\bigr\},
\end{align*}
where $\mathcal{Z}_{\mathrm{sup}}$ is the set of supported cuts, $\mathcal{Z}_{\mathrm{fix}}$ is the set of cuts that are PGA fixed points of the adjacency formulation under some strictly positive preference vector, and $\mathcal{Z}_{\mathrm{loc}}$ is the set that consists of cuts $s$ for which no $1$-bit-flip neighbor $s^{(v)}$ Pareto dominates $s$.

\begin{proposition}[From Global Support to Local Pareto Stability]
\label{prop:hierarchy}
Let $\pset$ be the Pareto-optimal cut set. Then, $\mathcal{Z}_{\mathrm{sup}}\subseteq\mathcal{Z}_{\mathrm{fix}}\subseteq\mathcal{Z}_{\mathrm{loc}}$ and $\mathcal{Z}_{\mathrm{sup}}\subseteq\pset\subseteq\mathcal{Z}_{\mathrm{loc}}$.
\end{proposition}

Proposition~\ref{prop:hierarchy} clarifies how MO-QUCO can explore candidates beyond exact WSM solutions. With strictly positive weights, exact WSM returns only supported cuts in $\mathcal{Z}_{\mathrm{sup}}$. In contrast, PGA may converge to binary fixed points in $\mathcal{Z}_{\mathrm{fix}}$ that need not belong to $\pset$. The subsequent $1$-bit-flip refinement explores neighboring points and retains non-dominated candidates. This refinement aims to discover additional Pareto-optimal cuts, including unsupported cuts, and improve coverage of the Pareto front.

\paragraph{Discussion: Connection to the Preference Budget.}
The size of $\Lambda^{A}(s)$ connects the fixed-point analysis to the preference budget $J$.
We sample each preference uniformly from $\Delta^K$, i.e., $\lambda$ follows $\operatorname{Dirichlet}(1,\ldots,1)$, whose density is constant over the simplex. For a fixed cut $s$, we define the probability that a sampled preference makes $s$ a PGA fixed point as
\[
p_s
:=
\Pr_{\lambda \sim \operatorname{Dirichlet}(1,\ldots,1)}
\bigl[\lambda \in \Lambda^{A}(s)\bigr].
\]
Thus,
$p_s$ measures the fraction of preference vectors in $\Delta^K$ for which $s$ is a binary fixed point of the adjacency formulation.
Our method draws $J$ preferences independently. Hence, the probability that at least one preference lies in $\Lambda^{A}(s)$ is $1-(1-p_s)^J$. Thus, increasing $J$ improves preference-space coverage and creates more opportunities for PGA to reach different fixed points.

\subsection{The Proposed Batched Algorithm}
In this subsection, we introduce our \ul{m}ulti-\ul{o}bjective \ul{qu}adratic box-\ul{c}onstrained \ul{o}ptimization (MO-QUCO) algorithm and describe its main components in detail.

\begin{definition}[Non-domination operator]
\label{def:nondom}
For any candidate set $\mathcal{C}$ of binary vectors in $\{0,1\}^{n}$, $\textsc{NonDom}(\mathcal{C})$ returns the vectors in $\mathcal{C}$ that are not Pareto dominated by any other vector in $\mathcal{C}$.
\end{definition}

\begin{definition}[\textsc{BitFlip} function]
\label{def:bitflip}
For any candidate set $\hat{\mathcal{Z}}$ of binary vectors in $\{0,1\}^{n}$, the function $\textsc{BitFlip}(\hat{\mathcal{Z}})$ returns all $1$-bit-flip neighbors of these vectors:
\[
\textsc{BitFlip}(\hat{\mathcal{Z}})
:=
\{z^{(v)}:z\in\hat{\mathcal{Z}},\ v\in[n]\},
\]
where $z^{(v)}$ is obtained by replacing $z_v$ with $1-z_v$.
\end{definition}

Algorithm~\ref{alg:main} presents the full procedure of MO-QUCO.
In Step~\ref{alg:main:pref-loop}, MO-QUCO iterates over the $J$ preference vectors, where for each preference $\lambda_j$, we first aggregate the layer weight matrices into $W_{\lambda_j}$ by linear scalarization (Step~\ref{alg:main:aggregate}), and draw a batch of $B$ initial points uniformly from the box (Step~\ref{alg:main:sample}).
\begin{algorithm}[t]
\caption{The MO-QUCO algorithm.}
\label{alg:main}
\begin{algorithmic}[1]
\REQUIRE Graphs $\{G^{(k)} = (V, E, W^{(k)})\}_{k \in [K]}$, batch size $B$, PGA steps $T$, step size $\alpha$, number of preferences $J$, maximum number of bit-flip passes $P$
\STATE \textbf{Initialization:} candidate set $\hat{\mathcal{Z}} \gets \emptyset$, sample $\{\lambda_j\}_{j=1}^{J}$ i.i.d.\ uniformly from $\Delta^K$
\FOR{$j = 1$ to $J$}\label{alg:main:pref-loop}
    \STATE $W_{\lambda_j} \gets \sum_{k \in [K]} \lambda_{j,k}\, W^{(k)}$\label{alg:main:aggregate}
    \STATE sample $X = [x_1^\top;\ldots;x_B^\top] \sim \mathcal{U}[-1,1]^{B \times n}$\label{alg:main:sample}
    \FOR{$t = 1$ to $T$}\label{alg:main:pga-start}
        \STATE $X \gets \Proj_{[-1,1]}\!\bigl(X - \alpha\, X W_{\lambda_j}\bigr)$ \COMMENT{PGA on $g^{A}_{\lambda_j}$}
    \ENDFOR\label{alg:main:pga-end}
    \STATE $\mathcal{Z}_j \gets \{\ind\{x_b > 0\} : b \in [B]\}$\label{alg:main:binarize}
    \STATE $\hat{\mathcal{Z}} \gets \hat{\mathcal{Z}} \cup \mathcal{Z}_j$\label{alg:main:collect}
\ENDFOR
\STATE $\hat{\mathcal{Z}} \gets \textsc{NonDom}\bigl(\hat{\mathcal{Z}}\bigr)$\label{alg:main:nondom}
\FOR{$p = 1$ to $P$}\label{alg:main:bitflip-start}
    \STATE $\hat{\mathcal{Z}} \gets \textsc{NonDom}\!\left(\hat{\mathcal{Z}} \cup \textsc{BitFlip}(\hat{\mathcal{Z}})\right)$
    \STATE \textbf{break if} no new cut is found
\ENDFOR\label{alg:main:bitflip-end}
\STATE \textbf{return} $\hat{\mathcal{Z}}$
\end{algorithmic}
\end{algorithm}

\begin{figure}[t]
\centering
\includegraphics[width=0.94\textwidth]{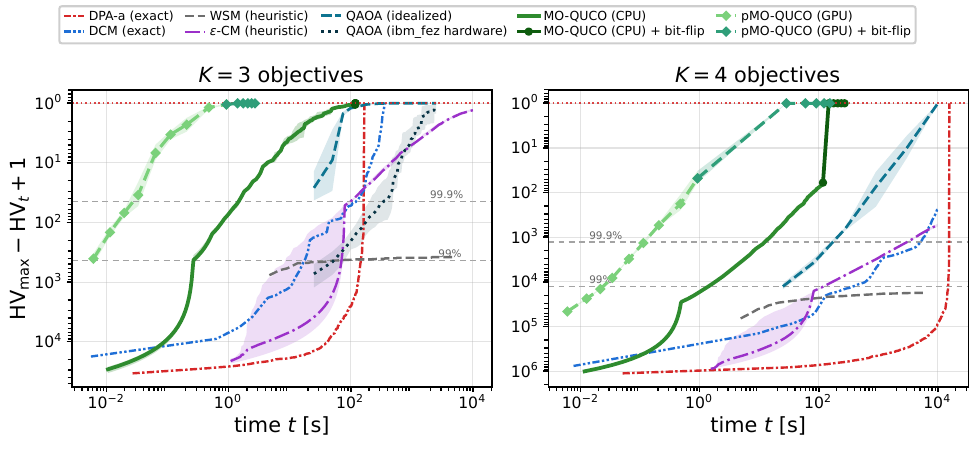}
\caption{Convergence on the 42-node MO-MaxCut benchmark with $K=3$ (left) and $K=4$ (right) objectives. The x-axis shows time on a logarithmic scale. The y-axis shows $\mathrm{HV}_{\max}-\mathrm{HV}_t+1$ on a reversed logarithmic scale, so higher is better. All classical methods use measured wall-clock time, while QAOA uses the benchmark's idealized sampling-time estimate at $10$\,kHz. The gray dashed lines correspond to $99\%$ and $99.9\%$ of $\mathrm{HV}_{\max}$. Shaded bands show the min-to-max range over five repetitions for QAOA and three seeds for the other stochastic methods. The $1$-bit-flip steps in the MO-QUCO curves use darker colors and markers.
}
\label{fig:main}
\end{figure}
In Steps~\ref{alg:main:pga-start} to~\ref{alg:main:pga-end}, we run $T$ PGA steps using the adjacency formulation $g^{A}_{\lambda}$ on the whole batch. Each iteration consists of one matrix product followed by an entry-wise projection. The resulting points are binarized in Step~\ref{alg:main:binarize} and added to $\hat{\mathcal{Z}}$ in Step~\ref{alg:main:collect}. After sweeping all $J$ preferences, Step~\ref{alg:main:nondom} retains only the non-dominated cuts using the $\textsc{NonDom}(\cdot)$ function. Steps~\ref{alg:main:bitflip-start} to~\ref{alg:main:bitflip-end} generate the $1$-bit-flip neighbors using the \textsc{BitFlip} function and apply the \textsc{NonDom} function after every pass until no new cut is found or $P$ passes have been completed.
\looseness=-1 %

The CPU implementation of Algorithm~\ref{alg:main} stacks the $B$ random initializations into matrix $X$ and updates the entire batch with one matrix multiplication at each PGA iteration.

The supplementary material presents pMO-QUCO (Appendix~\ref{app:pmo-quco}), which follows Algorithm~\ref{alg:main} and partitions the preference vectors into blocks. Within each block, the GPU accelerates linear scalarization in Step~\ref{alg:main:aggregate}, sampling in Step~\ref{alg:main:sample}, binarization in Step~\ref{alg:main:binarize}, and the batched matrix multiplication in Steps~\ref{alg:main:pga-start} to~\ref{alg:main:pga-end}.

\section{Experimental Results}

We evaluate our algorithm on the MO-MaxCut benchmark of \citet{kotil2025quantum}, which contains two instances with $K=3$ and $K=4$ objectives on a shared graph of $42$ nodes and $46$ edges. The layer weights are sampled i.i.d.\ from $\mathcal{N}(0,1)$.
To assess performance on larger graphs, we also generate Erd\H{o}s--R\'enyi (ER) graphs with $p=0.5$ and complete graphs for $n\in\{100,200\}$.
For each size and graph topology, we generate five independent instances with $K=3$ structured conflicting objectives.
Due to space limitations, we provide the generation rules in Appendix~\ref{app:conflict} and the complete-graph results in Appendix~\ref{app:anytime-large}.

\begin{figure}[t]
\centering
\includegraphics[width=0.94\textwidth]{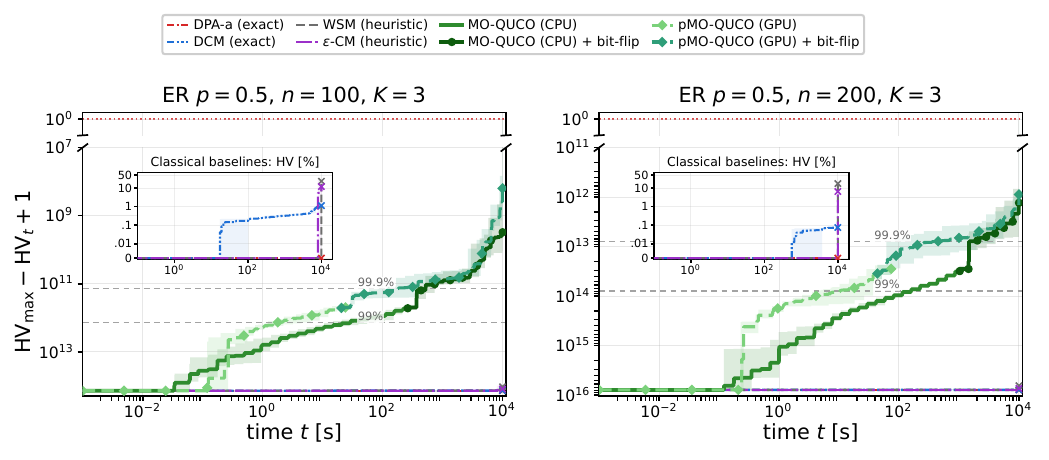}
\caption{Convergence on generated ER graphs with $n=100$ (left) and $n=200$ (right), both with $K=3$.
Both axes are logarithmic and the HV-gap axis is reversed.
Lines and bands show the mean and min-to-max range over five graphs.
The gray dashed lines correspond to $99\%$ and $99.9\%$ of the mean per-instance $\mathrm{HV}_{\max}$ in each panel.
Slashes omit empty gap ranges.
Insets show classical-baseline HV percentages.
}
\label{fig:larger-main}
\end{figure}

\paragraph{1) Baselines:}
We compare four classical solvers evaluated by \citet{kotil2025quantum}.
The first two are the exact multi-objective integer-programming methods DPA-a \citep{dachert2024simple} and DCM \citep{boland2017new}, which enumerate the non-dominated set with CPLEX \citep{manual1987ibm}.
The other two are the randomized heuristics $\varepsilon$-CM \citep{haimes1971bicriterion} and WSM \citep{zadeh1963optimality}, which solve sampled $\varepsilon$-constraint and weighted-sum programs with Gurobi \citep{gurobi2023gurobi}.
We also compare with QAOA \citep{kotil2025quantum} at its best reported configuration, using the depth-$6$ matrix-product-state simulation results distributed with the benchmark and the \texttt{ibm\_fez} hardware runs available for $K=3$.
We run the classical baselines and MO-QUCO on the same server, and we adopt the QAOA simulation and \texttt{ibm\_fez} hardware results as reported by the benchmark.
For runtime evaluation, both the CPU and GPU implementations use measured wall-clock time. The QAOA runtime follows the idealized estimate of \citet{kotil2025quantum}. It is computed by dividing the cumulative number of circuit shots by $10{,}000$ shots per second.
\footnote{QAOA time is an idealized hardware sampling estimate based on a $10^{-4}$\,s repetition delay and excludes circuit compilation.}

\paragraph{2) Metrics:}
To compare algorithms in the multi-objective setting, we report the number of non-dominated points (NDP, cf.\ Definition~\ref{def:pareto}) in the final archive and
its hypervolume (HV) \citep{zitzler1999multiobjective, yang2019efficient}.
For this maximization problem, fix a reference point
$r\in\R^K$ such that $r_k\leq f_k(z)$ for every feasible cut $z$ and every $k\in[K]$. For each feasible cut $z$, let
$[r,F(z)]$ denote the axis aligned box whose lower corner is $r$
and upper corner is $F(z)$.
The hypervolume of an archive $\hat{\mathcal{Z}}$ is defined as
\begin{equation}
\mathrm{HV}_r(\hat{\mathcal{Z}})
:=\operatorname{vol}_K(\cup_{z\in\hat{\mathcal{Z}}}[r,F(z)]),
\label{eq:hv}
\end{equation}
where $\operatorname{vol}_K$ denotes the $K$ dimensional volume.
For a fixed reference point, a larger HV means
that the archive dominates a larger hypervolume (i.e., region of the objective space).

For the 42-node benchmark, we set $r$ to the per-objective lower bounds provided by the benchmark and compute HV with the \texttt{pygmo} library \citep{biscani2020parallel}. In the tables, we report $\%\mathrm{HV}_{\max}
:=
100\,\frac{\mathrm{HV}_r(\hat{\mathcal{Z}})}{\mathrm{HV}_{\max}},$
where $\mathrm{HV}_{\max}$ is the hypervolume of the benchmark's exact non-dominated front evaluated using the same reference point.
For the generated larger instances, $\mathrm{HV}_{\max}$ instead denotes the best-observed final HV among the displayed methods on each instance.
\paragraph{3) Main Experimental Results:}
Figures~\ref{fig:main} and~\ref{fig:larger-main} compare anytime performance across the original benchmark and the larger generated instances.
The figures illustrate that MO-QUCO achieves rapid initial improvement with PGA, benefits from GPU acceleration, and further improves solution quality through $1$-bit-flip refinement.

\noindent\textbf{3-a) 42-Node Benchmark.}
Figure~\ref{fig:main} shows that MO-QUCO reduces the HV gap rapidly for both $K=3$ and $K=4$.
GPU batching shifts this improvement to earlier times.
pMO-QUCO reaches more than $99.99\%$ of $\mathrm{HV}_{\max}$ in about $0.9$\,s on both instances.
Figure~\ref{fig:main} also shows that the subsequent $1$-bit-flip refinement closes the remaining numerical HV gap.
The classical and QAOA baselines take longer to attain comparable HV or retain larger gaps.
This separation is particularly clear for $K=4$.
Detailed numerical results are given in Table~\ref{tab:main} in Appendix~\ref{app:benchmark42}.

\noindent\textbf{3-b) Larger Instances.}
Figure~\ref{fig:larger-main} shows the same pattern on ER graphs with $n=100,200$.
At both graph sizes, MO-QUCO and pMO-QUCO achieve over $99\%$ of $\mathrm{HV}_{\max}$ on average without refinement and over $99.9\%$ with $1$-bit-flip refinement.
In contrast, as shown in the insets, the classical baselines make limited progress and even the best achieves less than $50\%$ of $\mathrm{HV}_{\max}$ on average.
Complete-graph results show a similar pattern (Figure~\ref{fig:anytime-large} in Appendix~\ref{app:anytime-large}).

\paragraph{4) Additional Experimental Findings:}
Due to space limitation, we provide ablation studies in the supplementary materials, which show that PGA provides the main improvement and the $1$-bit-flip refinement is essential (Table~\ref{tab:components}).
Furthermore, the sensitivity analysis in Figure~\ref{fig:ablation} shows that more sampled points improve solution quality.
The number of PGA steps $T$ and the step size $\alpha$ should be tuned to obtain better solutions.
In addition, the formulation comparison supports the adjacency formulation as a balance between quality and runtime (Figure~\ref{fig:forms} and Table~\ref{tab:forms}).
Overall, experiments on both small and large graphs show that our methods achieve higher HV and run faster than the competing baselines.

\section{Conclusions}

In this paper, we developed a differentiable approach for the multi-objective maximum cut (MaxCut) problem. We used linear scalarization to convert the multi-objective problem to a single-objective signed-weight MaxCut problem, enabling the use of gradient-based optimization methods. For a fixed preference vector, we characterized the fixed points of the resulting optimization problem under the adopted adjacency-based formulation. Furthermore, we analyzed the local dynamics near PGA fixed points.
We then analyzed the preference-conditioned fixed points associated with a given binary solution, characterizing the set of preferences for which the solution remains a fixed point and establishing their connection to the Pareto properties. Experimentally, across different numbers of objectives and graph sizes, we demonstrated that our algorithms achieve state-of-the-art performance compared with exact and heuristic-based methods. %

\subsection*{AI use statement}

Generative AI tools were used to polish the writing, improve the presentation of figures, and check portions of mathematical derivations and proofs.
The authors independently verified all AI-assisted text, mathematical claims, proofs, citations, and code.
The authors take full responsibility for the content of this paper.

\bibliography{aaai2027/references}

\clearpage
\appendix
\setcounter{secnumdepth}{3}
\par\noindent\rule{\textwidth}{1pt}
\begin{center}
{\Large \bfseries Appendix}
\end{center}
\par\noindent\rule{\textwidth}{1pt}

\section{Proofs}\label{app sec: proof}

\subsection{Proof of Lemma~\ref{lem:binary-equivalence}}
\label{app:binary-equivalence}
\begin{proof}
For every $s\in\{-1,1\}^n$, we have $s^\top s=n$ and $s^\top D_{W_\lambda}s=\operatorname{tr}(D_{W_\lambda})$. Since $L_\lambda=D_{W_\lambda}-W_\lambda$,
\[
g^P_\lambda(s)=g^L_\lambda(s)+\frac{\delta n}{4},
\qquad
g^A_\lambda(s)=g^L_\lambda(s)-\frac14\operatorname{tr}(D_{W_\lambda}).
\]
Both offsets are independent of $s$, so these three objectives have the same binary maximizers. In contrast,
\[
g^B_\lambda(s)=g^A_\lambda(s)-\frac{\delta}{4}\onevec_n^\top s.
\]
The additional term depends on $s$ and has magnitude at most $\delta n/4$. It can change the binary maximizers.
\end{proof}

\subsection{Proof of Lemma~\ref{lem:fixedpoint}}
\begin{proof}
The projection onto $[-1,1]^{n}$ clips each coordinate separately. Fix a coordinate $v$. Under Equation~\eqref{eq:pga} in the main paper, $x_v$ is mapped to
\[
\min\{1,\max\{-1,\,x_v\ -\alpha(W_\lambda x)_v\}\}.
\]
Thus, $x$ is a PGA fixed point iff every coordinate is mapped to itself. We consider three cases.

\emph{Case 1: $(\ W_\lambda x)_v=0$.}
In this case, the updated coordinate equals $x_v$, so coordinate $v$ is unchanged for every $x_v\in[-1,1]$.

\emph{Case 2: $(\ W_\lambda x)_v>0$.}
Since $\alpha>0$, we have $\ x_v-\alpha(W_\lambda x)_v<x_v$. If $x_v\ >-1$, clipping returns either a value strictly between $\ -1$ and $\ x_v$ or the \ lower endpoint $\ -1$, and hence the updated coordinate is strictly \ less than $x_v$. Therefore, coordinate $v$ is unchanged iff $x_v=\ -1$. Equivalently,
\[
x_v=-\operatorname{sign}\!\bigl((\ W_\lambda x)_v\bigr).
\]

\emph{Case 3: $(\ W_\lambda x)_v<0$.}
Now $\ x_v-\alpha(W_\lambda x)_v>x_v$. If $x_v\ <1$, clipping returns either a value strictly between $\ x_v$ and $\ 1$ or the \ upper endpoint $\ 1$, and hence the updated coordinate is strictly \ greater than $x_v$. Therefore, coordinate $v$ is unchanged iff $x_v=\ 1$, which again is equivalent to
\[
x_v=-\operatorname{sign}\!\bigl((\ W_\lambda x)_v\bigr).
\]

Combining the three cases, coordinate $v$ is unchanged iff either $(\ W_\lambda x)_v=0$ or $x_v=-\operatorname{sign}\!\bigl((\ W_\lambda x)_v\bigr)$. Hence, $x$ is a PGA fixed point iff
\[
x_v=-\operatorname{sign}\!\bigl((\ W_\lambda x)_v\bigr)
\]
for every coordinate $v$ with $(\ W_\lambda x)_v\neq0$. \ Points with $W_\lambda x=0$ also satisfy the condition because every coordinate is unchanged. Since the argument holds for every $\alpha>0$, the result follows.
\end{proof}

\subsection{Proof of Theorem~\ref{thm:binaryfixed}}

\begin{proof}
Fix a vertex $v \in [n]$. \ The $v$-th coordinate of $\ W_\lambda s$ is
\[
\qquad
(W_\lambda s)_v=\sum_u(W_\lambda)_{uv}s_u.
\]
The matrix $D_{W_\lambda}$ is diagonal with
$
(D_{W_\lambda})_{vv}
=
\sum_u(W_\lambda)_{uv}=d_v.
$
Therefore,
\[
(D_{W_\lambda}s)_v
=
\left(\sum_u(W_\lambda)_{uv}\right)s_v.
\]
Since $W_\lambda$ is symmetric, matrix-vector multiplication gives
$
(W_\lambda s)_v
=
\sum_u(W_\lambda)_{uv}s_u.
$
So
\[
\begin{aligned}
(\ W_\lambda s)_v
&=
\sum_u(W_\lambda)_{uv}s_u\\
&=
\sum_u(W_\lambda)_{uv}\ s_u.
\end{aligned}
\]

We now partition the sum according to whether $u$ and $v$ lie on the same side of the cut:
\[
\begin{aligned}
(\ W_\lambda s)_v
&=
\sum_{u:\,s_u=s_v}(W_\lambda)_{uv}\ s_u\\
&\quad+
\sum_{u:\,s_u\neq s_v}(W_\lambda)_{uv}\ s_u.
\end{aligned}
\]
The first sum is \ $s_v\sum_{u:\,s_u=s_v}(W_\lambda)_{uv}$ because $s_u=s_v$. In the second sum, $s_u=-s_v$ because $s_u,s_v\in\{-1,1\}$. Hence, \ the two sums give
\[
\begin{aligned}
(\ W_\lambda s)_v
&=
\quad
s_v\sum_{u:\,s_u=s_v}(W_\lambda)_{uv}-s_v\sum_{u:\,s_u\neq s_v}(W_\lambda)_{uv}\\
&=
\quad s_v\bigl(d_v-2C_v(s)\bigr).
\end{aligned}
\]
Multiplying by $s_v$ and using $s_v^2=1$ gives the identity
\[
s_v(\ W_\lambda s)_v
=
\quad d_v-2C_v(s).
\]

Next, we apply Lemma~\ref{lem:fixedpoint} in the main paper. If $(\ W_\lambda s)_v\neq0$, coordinate $v$ is unchanged iff $s_v$ has the \ opposite sign as $(\ W_\lambda s)_v$. Since $s_v\in\{-1,1\}$, this condition is equivalent to
\[
s_v(\ W_\lambda s)_v\ <0.
\]
If $(\ W_\lambda s)_v=0$, coordinate $v$ is also unchanged and $s_v(\ W_\lambda s)_v=0$. Combining the two cases, coordinate $v$ is unchanged iff
\[
s_v(\ W_\lambda s)_v\ \leq0.
\]
Since $s_v(\ W_\lambda s)_v=\ d_v-2C_v(s)$, this condition is equivalent to $C_v(s)\geq\ d_v/2$. The full vector $s$ is unchanged iff every coordinate is unchanged. Therefore, $s$ is a PGA fixed point iff $C_v(s)\geq\ d_v/2$ for every $v\in[n]$. \ This condition also includes binary stationary points with $W_\lambda s=0$.

Only edges incident to $v$ change their cut status after a flip. Edges on the same side become cut and contribute $(W_\lambda)_{uv}$. Crossing edges become uncut and contribute $-(W_\lambda)_{uv}$. Hence,
\[
\Delta_v(s)
=\sum_{u:\,s_u=s_v}(W_\lambda)_{uv}
-\sum_{u:\,s_u\neq s_v}(W_\lambda)_{uv}
=s_v(W_\lambda s)_v=d_v-2C_v(s).
\]
Thus, the cut admits no improving $1$-bit flip if and only if $\Delta_v(s)\leq0$ for every $v$. This is equivalent to the fixed-point condition proved above.

\end{proof}

\subsection{Proof of Theorem~\ref{thm:adjacency}}
\label{app:adjacency-dynamics}
\begin{proof}
By Theorem~\ref{thm:binaryfixed}, the flip gain satisfies $\Delta_v(s)=s_v(W_\lambda s)_v$. We use this identity to analyze the local dynamics below.

{
\paragraph{(i) Convergence near binary fixed points.}
Fix $\lambda\in\Delta^K$ and $\alpha>0$. Suppose $\Delta_v(s)<0$ for every $v\in[n]$. Define
\[
\mu:=\min_{v\in[n]}\{-\Delta_v(s)\}>0,
\qquad
M:=\max_{v\in[n]}\sum_{u\in[n]}|(W_\lambda)_{vu}|,
\qquad
\rho:=\frac{\mu}{2M}.
\]
More precisely, every $x^0\in[-1,1]^n$ with $\|x^0-s\|_\infty\leq\rho$ satisfies
\begin{equation}
\|x^t-s\|_\infty
\leq\max\left\{0,\|x^0-s\|_\infty-\frac{t\alpha\mu}{2}\right\},
\qquad t\geq0,
\label{eq:adjacency-local-bound}
\end{equation}
under \eqref{eq:pga}. Hence the PGA iterates reach $s$ within at most
\[
\left\lceil\frac{2\|x^0-s\|_\infty}{\alpha\mu}\right\rceil
\]
updates.

To prove the bound, the flip-gain identity gives $s_v(W_\lambda s)_v=\Delta_v(s)\leq-\mu$. Moreover,
\[
\mu\leq|(W_\lambda s)_v|
\leq\sum_u|(W_\lambda)_{vu}|\leq M.
\]
Thus $M>0$ and $0<\rho\leq1/2$. For any iterate $x^t\in[-1,1]^n$ with $\|x^t-s\|_\infty\leq\rho$,
\[
\begin{aligned}
s_v(W_\lambda x^t)_v
&=s_v(W_\lambda s)_v+s_v\bigl(W_\lambda(x^t-s)\bigr)_v\\
&\leq-\mu+\sum_u|(W_\lambda)_{vu}|\,|x_u^t-s_u|\\
&\leq-\mu+M\rho=-\frac{\mu}{2}.
\end{aligned}
\]
The coordinate-wise projection therefore satisfies
\[
s_vx_v^{t+1}\geq\min\left\{1,s_vx_v^t+\frac{\alpha\mu}{2}\right\}.
\]
Since $s_v\in\{-1,1\}$ and $x_v^t\in[-1,1]$, this implies
\[
|x_v^{t+1}-s_v|\leq\max\left\{0,|x_v^t-s_v|-\frac{\alpha\mu}{2}\right\}.
\]
Taking the maximum over $v$ gives
\[
\|x^{t+1}-s\|_\infty
\leq\max\left\{0,\|x^t-s\|_\infty-\frac{\alpha\mu}{2}\right\}
\leq\|x^t-s\|_\infty\leq\rho.
\]
Thus every PGA iterate remains in $[-1,1]^n$ and within distance $\rho$ of $s$. Induction yields \eqref{eq:adjacency-local-bound} and the stated finite-step bound for every $\alpha>0$.

\paragraph{(ii) Nonbinary fixed points.}
We use stability in the following sense. For every $\eta>0$, there must be a $\delta>0$ such that
\[
x^0\in[-1,1]^n,\quad \|x^0-x^*\|_\infty<\delta
\quad\Longrightarrow\quad
\|x^t-x^*\|_\infty<\eta\quad\text{for all }t\geq0.
\]
Local asymptotic stability additionally requires $x^t\to x^*$ from all sufficiently close feasible initial points.

Let $x^*$ be a fixed point of \eqref{eq:pga} satisfying the assumptions in part~(ii). Set $\mathcal A:=[n]\setminus\mathcal I$. For each $v\in\mathcal I$, the coordinate is strictly inside the box. Therefore,
\[
x_v^*=x_v^*-\alpha(W_\lambda x^*)_v,
\qquad\text{so}\qquad (W_\lambda x^*)_{\mathcal I}=0.
\]
The inequalities $x_v^*(W_\lambda x^*)_v<0$ for $v\in\mathcal A$ and $|x_v^*|<1$ for $v\in\mathcal I$ imply that there is a neighborhood $U$ of $x^*$ such that
\[
x^t\in U,\quad x^t_{\mathcal A}=x^*_{\mathcal A}
\quad\Longrightarrow\quad
\begin{cases}
x^{t+1}_{\mathcal A}=x^*_{\mathcal A},\\
x^{t+1}_{\mathcal I}=x^t_{\mathcal I}-\alpha(W_\lambda x^t)_{\mathcal I}.
\end{cases}
\]
Set $e^t:=x^t_{\mathcal I}-x^*_{\mathcal I}$ and $H:=(W_\lambda)_{\mathcal I\mathcal I}$. With $x^t_{\mathcal A}=x^*_{\mathcal A}$, the update of the interior coordinates is
\begin{equation}
e^{t+1}:=x^{t+1}_{\mathcal I}-x^*_{\mathcal I}
=\bigl(I_{|\mathcal I|}-\alpha H\bigr)e^t,
\label{eq:adjacency-face-dynamics}
\end{equation}
where $I_{|\mathcal I|}$ denotes the identity matrix. Each $W^{(k)}$ is symmetric with zero diagonal. Hence $H=H^\top$, $\operatorname{diag}(H)=0$ and $\operatorname{tr}(H)=0$.

If $H\neq0$, symmetry and zero trace imply that $H$ has a negative eigenvalue $\theta$. Choose an eigenvector $q$ with $Hq=\theta q$ and $\|q\|_\infty=1$. Then $\beta:=1-\alpha\theta>1$. Choose $r>0$ small enough that
\[
y_{\mathcal A}=x^*_{\mathcal A},\quad
\|y_{\mathcal I}-x^*_{\mathcal I}\|_\infty\leq r
\quad\Longrightarrow\quad y\in U\cap[-1,1]^n.
\]
For any $0<\varepsilon<r$, choose $x^0_{\mathcal A}=x^*_{\mathcal A}$ and $e^0=\varepsilon q$. Define
\[
N:=\min\{t\geq0:\beta^t\varepsilon>r\}.
\]
For $0\leq t<N$, induction using \eqref{eq:adjacency-face-dynamics} gives
\[
e^t=\beta^t\varepsilon q,\qquad
\|e^t\|_\infty=\beta^t\varepsilon\leq r,\qquad
e^{t+1}=\beta e^t.
\]
This relation also gives the update from $t=N-1$ to $t=N$. In particular, $\|e^N\|_\infty=\beta^N\varepsilon>r$. The iterates therefore leave the neighborhood of radius $r$ even when the initial distance $\varepsilon$ is arbitrarily small. Thus $x^*$ is unstable.

If $H=0$, \eqref{eq:adjacency-face-dynamics} gives $e^{t+1}=e^t$. Every sufficiently small nonzero change to the interior coordinates defines another fixed point $y\neq x^*$ when the boundary coordinates are unchanged. PGA iterates starting at such a $y$ remain at $y$ and do not converge to $x^*$. This does not rule out convergence to the set of nearby fixed points. In either case, $x^*$ is not locally asymptotically stable.

Finally, \eqref{eq:scalarized-flip-gain} gives $\Delta_v(s)=\lambda^\top\gamma_v(s)$. Any preference satisfying $\lambda^\top\gamma_v(s)<0$ for every $v$ therefore makes nearby PGA iterates reach $s$ in finitely many steps by part~(i). The weak inequalities defining $\Lambda^A(s)$ guarantee that $s$ is a fixed point but do not guarantee convergence from nearby initial points.
}
\end{proof}

\subsection{Proof of Lemma~\ref{lem:supported-pareto}}
\begin{proof}
Let $z^\star$ be supported. By Definition~\ref{def:supported-outcomes} in the main paper, there exists $\lambda\in\Delta_{+}^{K}$ such that $z^\star$ globally maximizes $h_\lambda$. Suppose that a cut $z$ dominates $z^\star$. Dominance gives
\[
f_k(z)\geq f_k(z^\star)
\qquad\text{for every }k\in[K].
\]
Moreover, the inequality is strict for at least one objective. Since $\lambda_k>0$ for every $k$, we have
\[
h_\lambda(z)-h_\lambda(z^\star)
=
\sum_{k=1}^{K}\lambda_k
\bigl(f_k(z)-f_k(z^\star)\bigr)
>0.
\]
This contradicts the global optimality of $z^\star$. Hence no cut dominates $z^\star$. Therefore, $z^\star$ is Pareto-optimal.
\end{proof}

\subsection{Proof of Theorem~\ref{thm:prefpolytope}}
\begin{proof}
By Theorem~\ \ref{thm:binaryfixed} in the main paper, condition~(i) holds if and only if
\[
s_v(W_\lambda s)_v\leq0
\qquad\text{for every }v\in[n].
\]
Using $W_\lambda=\sum_{k=1}^K\lambda_kW^{(k)}$ and the definition of the flip-gain vector, we obtain
\[
\begin{aligned}
s_v(W_\lambda s)_v
&=
\sum_{k=1}^K\lambda_k
s_v\bigl(W^{(k)}s\bigr)_v\\
&=
\sum_{k=1}^K\lambda_k\Delta_v^{(k)}(s)\\
&=
\lambda^\top\gamma_v(s).
\end{aligned}
\]
Therefore, condition~(i) holds if and only if $\lambda^\top\gamma_v(s)\leq0$ for every $v\in[n]$, which is condition~(ii). By Definition~\ref{def:prefpolytope} in the main paper, condition~(ii) holds exactly when $\lambda\in\Lambda^{A}(s)$, which is condition~(iii). Therefore, all three conditions are equivalent.
\end{proof}

\subsection{Proof of Proposition~\ref{prop:hierarchy}}

\begin{proof}
\emph{First inclusion.}
Let $s\in\mathcal{Z}_{\mathrm{sup}}$. By definition, there exists $\lambda\in\Delta_+^K$ such that $s$ globally maximizes $h_\lambda$. For every $v\in[n]$, the neighboring cut $s^{(v)}$ is feasible. Hence,
\[
\lambda^\top\gamma_v(s)
=
h_\lambda\bigl(s^{(v)}\bigr)-h_\lambda(s)
\leq0.
\]
It follows that $\lambda\in\Lambda^A(s)$. Since $\lambda\in\Delta_+^K$, we have $\Lambda^A(s)\cap\Delta_+^K\neq\emptyset$. Therefore, $s\in\mathcal{Z}_{\mathrm{fix}}$, which proves
\[
\mathcal{Z}_{\mathrm{sup}}
\subseteq
\mathcal{Z}_{\mathrm{fix}}.
\]

\emph{Second inclusion.}
Let $s\in\mathcal{Z}_{\mathrm{fix}}$. Then, there exists $\lambda\in\Lambda^A(s)\cap\Delta_+^K$. Suppose for contradiction that $s\notin\mathcal{Z}_{\mathrm{loc}}$. There must exist a vertex $v\in[n]$ such that $s^{(v)}$ Pareto dominates $s$. Thus,
\[
\Delta_v^{(k)}(s)\geq0
\quad\text{for every }k\in[K],
\]
and at least one of these inequalities is strict. Since $\lambda_k>0$ for every $k\in[K]$, we obtain
\[
\lambda^\top\gamma_v(s)
=
\sum_{k=1}^K
\lambda_k\Delta_v^{(k)}(s)
>0.
\]
This contradicts $\lambda\in\Lambda^A(s)$, which requires $\lambda^\top\gamma_v(s)\leq0$ for every $v\in[n]$. Hence, $s\in\mathcal{Z}_{\mathrm{loc}}$, and therefore
\[
\mathcal{Z}_{\mathrm{fix}}
\subseteq
\mathcal{Z}_{\mathrm{loc}}.
\]

\emph{Pareto-optimal cuts.}
By Lemma~\ref{lem:supported-pareto}, $\mathcal{Z}_{\mathrm{sup}}\subseteq\pset$. Every Pareto-optimal cut is non-dominated among all feasible cuts and hence among its single-bit-flip neighbors. Therefore, $\pset\subseteq\mathcal{Z}_{\mathrm{loc}}$. Together with the first two inclusions, this proves the proposition.
\end{proof}

\section{Parallelized MO-QUCO (pMO-QUCO)}\label{app:pmo-quco}

Algorithm~\ref{alg:parallel} gives \emph{pMO-QUCO}, the GPU-parallelizable implementation.
\begin{algorithm}[H]
\caption{pMO-QUCO}
\label{alg:parallel}
\begin{algorithmic}[1]
\REQUIRE graphs $\{G^{(k)} = (V, E, W^{(k)})\}_{k \in [K]}$, batch size $B$, PGA steps $T$, step size $\alpha$, number of preferences $J$, block size $M$, maximum number of bit-flip passes $P$
\STATE \textbf{Initialization:} candidate set $\hat{\mathcal{Z}} \gets \emptyset$, sample $\{\lambda_j\}_{j=1}^{J}$ i.i.d.\ uniformly from $\Delta^K$, partition $[J]$ into blocks $\mathcal{I}_1,\ldots,\mathcal{I}_R$ of size $M$
\FOR{$r = 1$ to $R$}\label{alg:parallel:block-loop}
    \STATE $W_{\lambda_j} \gets \sum_{k \in [K]} \lambda_{j,k}\, W^{(k)}$ for all $j \in \mathcal{I}_r$ \COMMENT{GPU}\label{alg:parallel:aggregate}
    \STATE sample $X_j \sim \mathcal{U}[-1,1]^{B \times n}$ for all $j \in \mathcal{I}_r$ \COMMENT{GPU}\label{alg:parallel:sample}
    \FOR{$t = 1$ to $T$}\label{alg:parallel:pga-start}
        \STATE $X_j \gets \Proj_{[-1,1]}\!\bigl(X_j - \alpha\, X_j W_{\lambda_j}\bigr)$ for all $j \in \mathcal{I}_r$ \COMMENT{GPU: all $(j,b)$}
    \ENDFOR\label{alg:parallel:pga-end}
    \STATE $\mathcal{Z}_r \gets \{\ind\{x_{j,b} > 0\} : j \in \mathcal{I}_r,\ b \in [B]\}$ \COMMENT{GPU}\label{alg:parallel:binarize}
    \STATE $\hat{\mathcal{Z}} \gets \hat{\mathcal{Z}} \cup \mathcal{Z}_r$\label{alg:parallel:collect}
\ENDFOR
\STATE $\hat{\mathcal{Z}} \gets \textsc{NonDom}\bigl(\hat{\mathcal{Z}}\bigr)$ \COMMENT{CPU}\label{alg:parallel:nondom}
\FOR{$p = 1$ to $P$}\label{alg:parallel:bitflip-start}
    \STATE $\hat{\mathcal{Z}} \gets \textsc{NonDom}\!\left(\hat{\mathcal{Z}} \cup \textsc{BitFlip}(\hat{\mathcal{Z}})\right)$ \COMMENT{CPU}
    \STATE \textbf{break if} no new cut is found
\ENDFOR\label{alg:parallel:bitflip-end}
\STATE \textbf{return} $\hat{\mathcal{Z}}$
\end{algorithmic}
\end{algorithm}

\section{Additional Experimental Details and Results}
\label{app:experiments}
This section reports detailed results on the published 42-node benchmark and on generated conflicting-objective instances.

\subsection{Implementation Details}
\label{app:implementation}
\paragraph{Data and Instance Availability.}
\label{app:data}
Our experiments use the public QAMOO benchmark of \citet{kotil2025quantum} and an adapted version of the conflict-instance generator released with \citet{tao2026multi}.
The corresponding citations and generation procedures are provided in the paper and the supplementary code.
The QAMOO benchmark repository is released under the Apache License 2.0 and provides the instances, the QAOA results and the implementations of WSM and $\varepsilon$-CM.
The conflict-instance generator of \citet{tao2026multi} and the DPA-a implementation of \citet{dachert2024simple} are both released under the MIT License.
The \texttt{pygmo} library is released under the Mozilla Public License 2.0, and Gurobi and CPLEX are commercial solvers that we run under free academic licenses.

The implementation details below refer to the 42-node benchmark.
The larger-instance protocol is given in Appendix~\ref{app:anytime-large}.
\paragraph{MO-QUCO.}
Unless stated otherwise, we use $J = 100{,}000$ preferences drawn i.i.d.\ uniformly from $\Delta^K$, batch size $B = 64$, $T = 150$ PGA steps, and step size $\alpha = 0.05$. The bit-flip refinement runs for at most five passes and stops early once a pass adds no new nondominated cut.
Each stochastic run uses a fixed random seed from $\{0, 1, 2\}$. All CPU experiments run on a single core of an AMD EPYC 9555 CPU with the linear-algebra backend pinned to one thread, so runtimes are directly comparable.
\paragraph{pMO-QUCO.}
pMO-QUCO uses the same default parameters and random seeds as MO-QUCO. The preference sweep runs on one NVIDIA H200 NVL GPU with 140\,GiB of memory. The $1$-bit-flip refinement runs on one CPU core.
\paragraph{Baselines.}
Following the benchmark, we scale the layer weights by $10^{3}$ and round them because DPA-a and DCM require integer coefficients. All reported HV values are evaluated with the exact continuous weights. Under the $10{,}000$\,s stopping rule, $\varepsilon$-CM solves an average of $650{,}970$ scalarized subproblems for $K=3$ and $484{,}867$ for $K=4$ over three seeds. WSM solves $800{,}000$ randomly weighted scalarizations exactly. For the quantum baseline, we do not rerun the simulation. We use the depth-6 matrix-product-state sampling results distributed with the benchmark. We use the same bond dimensions as the benchmark's main results. Specifically, $\chi = 50$ for $K=3$ and $\chi = 20$ for $K=4$. Sampling times are estimated under an idealized $10$\,kHz shot rate \citep{kotil2025quantum}.
\FloatBarrier
\subsection{Detailed Results on the 42-Node Benchmark}
\label{app:benchmark42}
Table~\ref{tab:main} provides the final HV, archive size and runtime corresponding to Figure~\ref{fig:main}.
\input{aaai2027/arxiv_table_benchmark42.tex}

\FloatBarrier
\subsection{Ablation and Sensitivity Study}

\paragraph{Ablation.}
Table~\ref{tab:components} removes one search stage at a time under the same budget of $J = 10^{5}$ preferences and batch size $B = 64$. In the table,  PGA is the main driver of quality, which raises the random baseline from $78.74\%$ to $99.9999\%$ of $\mathrm{HV}_{\max}$ for $K=3$ and from $72.82\%$ to $99.9954\%$ for $K=4$. The bit-flip refinement then closes the last gap to the reference optimum. However, the bit-flip cannot stand alone, as it reaches only about $97\%$ on random cuts. In sum, both stages are necessary.

\begin{table}[!htbp]
\centering
\caption{Component ablation on the 42-node benchmark.
Results report the mean $\pm$ sample standard deviation over three seeds. A double dagger ($\ddagger$) marks an error within $2\times10^{-8}$.}
\label{tab:components}
{\small
\setlength{\tabcolsep}{4pt}
\begin{tabular}{lcc}
\toprule
& \multicolumn{2}{c}{\%$\mathrm{HV}_{\max}$} \\
\cmidrule(lr){2-3}
Configuration & $K=3$ & $K=4$ \\
\midrule
random & $78.74 \pm 0.31$ & $72.82 \pm 0.46$ \\
bit-flip & $97.00 \pm 0.60$ & $97.30 \pm 0.22$ \\
MO-QUCO & $99.9999 \pm 0.0001$ & $99.9954 \pm 0.0003$ \\
\;\; $+$ bit-flip & $\mathbf{100.0000 \pm 0.0000}$\rlap{$^{\ddagger}$} & $\mathbf{100.0000 \pm 0.0000}$\rlap{$^{\ddagger}$} \\
\bottomrule
\end{tabular}
}
\end{table}

\paragraph{Sensitivity.}
Figure~\ref{fig:ablation} shows that increasing the number of preferences $J$ or the batch size $B$ generally reduces the HV gap. Both parameters increase the number of sampled points. GPU acceleration makes it practical to increase them. The effects of the number of PGA steps $T$ and the step size $\alpha$ are also substantial and nonmonotonic.
For the number of PGA steps, $T=100$ gives the smallest observed gap before the bit-flip refinement. For the step size, $\alpha=0.02$ performs best for $K=3$ and $\alpha=0.05$ performs best for $K=4$. The reason is that longer runs and larger steps drive each point closer to a PGA fixed point, while  points that have not fully converged binarize into more diverse cuts and enrich the archive.
Thus, the number of sampled points controls the overall scaling trend, while $T$ and $\alpha$ affect the balance between optimization and archive diversity.
\begin{figure}[!htb]
\centering
\includegraphics[width=0.97\textwidth]{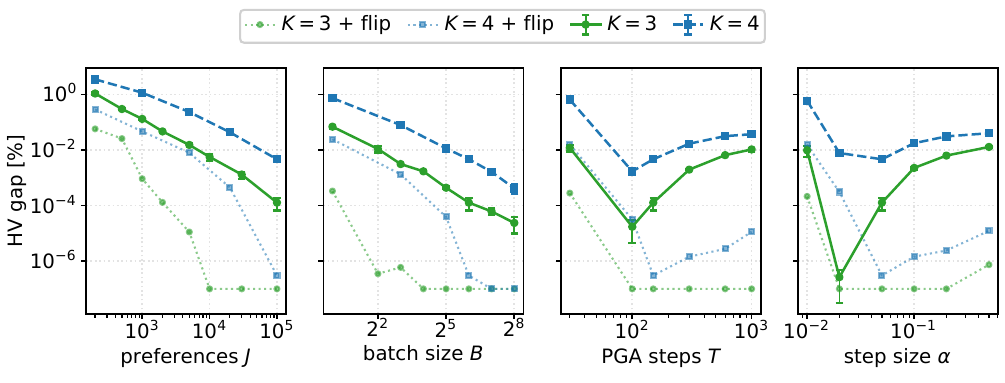}
\caption{Hyperparameter sensitivity of MO-QUCO under the default configuration of Table~\ref{tab:main} ($J=10^5$, $B=64$, $T=150$, $\alpha=0.05$). Each panel varies one hyperparameter. The vertical axis shows the percentage gap from $\mathrm{HV}_{\max}$. Curves and error bars show the mean $\pm$ one standard deviation over three seeds.}
\label{fig:ablation}
\end{figure}

\subsection{Choice of the Quadratic Formulation}
\label{app:quadratic-formulations}
We compare the four quadratic formulations under the default configuration of Table~\ref{tab:main}. The perturbed and biased formulations use $\delta=0.001$. Each formulation is evaluated across three seeds without the $1$-bit-flip refinement. Table~\ref{tab:forms} reports the HV gaps and single-core CPU solve times of the resulting archives.

\begin{table}[!htbp]
\centering
\caption{Quadratic formulation comparison on the 42-node benchmark under the default configuration without the $1$-bit-flip refinement. Results report the HV gap $\mathrm{HV}_{\max}-\mathrm{HV}$ and single-core CPU solve time as the mean $\pm$ sample standard deviation over three seeds. Lower is better for both metrics.}
\label{tab:forms}
{\small
\setlength{\tabcolsep}{2.8pt}
\begin{tabular}{lrrrr}
\toprule
 & \multicolumn{2}{c}{$K=3$} & \multicolumn{2}{c}{$K=4$} \\
\cmidrule(lr){2-3} \cmidrule(lr){4-5}
Formulation & Gap & Time (s) & Gap & Time (s) \\
\midrule
Laplacian $g^{L}_{\lambda}$ & 0.27\,$\pm$\,0.03 & 116.4\,$\pm$\,0.2 & 156\,$\pm$\,39 & 117.1\,$\pm$\,1.2 \\
Perturbed $g^{P}_{\lambda}$ & 0.28\,$\pm$\,0.01 & 116.8\,$\pm$\,0.1 & 157\,$\pm$\,37 & 119.1\,$\pm$\,3.6 \\
Adjacency $g^{A}_{\lambda}$ & 0.06\,$\pm$\,0.03 & 118.0\,$\pm$\,0.3 & 59\,$\pm$\,4 & 117.7\,$\pm$\,0.6 \\
Biased $g^{B}_{\lambda}$ & 0.06\,$\pm$\,0.03 & 140.5\,$\pm$\,2.3 & 57\,$\pm$\,4 & 139.4\,$\pm$\,1.0 \\
\bottomrule
\end{tabular}
}
\end{table}

\begin{figure}[!htb]
\centering
\includegraphics[width=\textwidth]{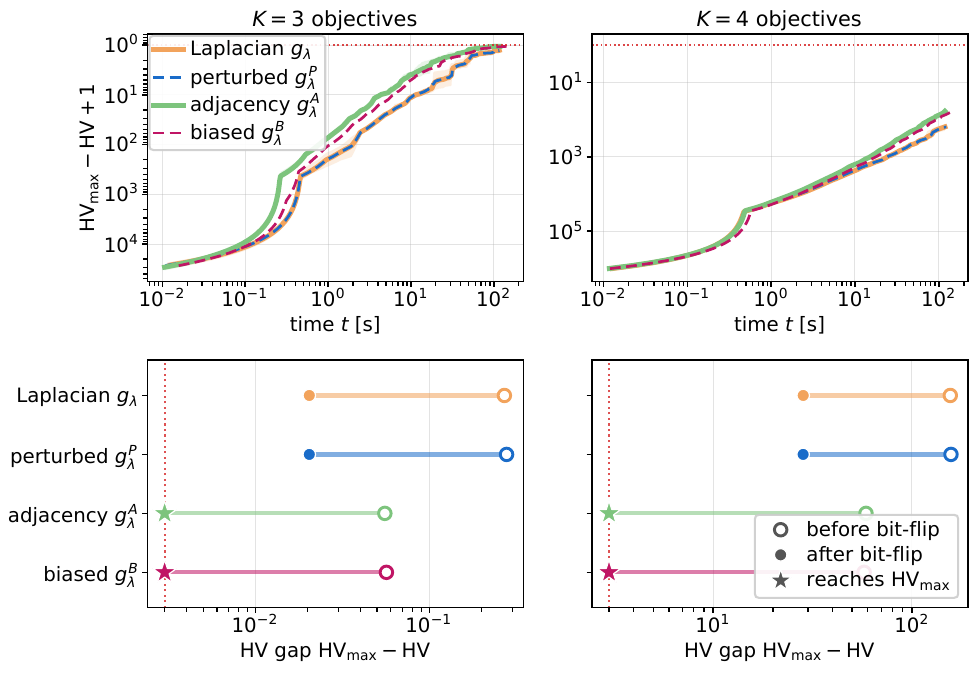}
\caption{Comparison of the four quadratic formulations on the 42-node benchmark under the default configuration. \textbf{Top:} $\mathrm{HV}_{\max}-\mathrm{HV}+1$ versus single-core CPU time without the $1$-bit-flip refinement. The vertical axis uses the reversed logarithmic scale of Figure~\ref{fig:main} in the main paper. Curves show the mean over three seeds, and shaded bands show the min-to-max range for the solid curves. \textbf{Bottom:} Mean HV gap $\mathrm{HV}_{\max}-\mathrm{HV}$ over three seeds before and after the $1$-bit-flip refinement. Open and filled markers denote the means before and after refinement, respectively.}
\label{fig:forms}
\end{figure}

Figure~\ref{fig:forms} shows that the adjacency formulation achieves  the smallest HV gap over time for both objective settings. The adjacency and biased formulations achieve similar final HV gaps before and after the $1$-bit-flip refinement. Their final gaps are smaller than those of the Laplacian and perturbed formulations. Table~\ref{tab:forms} shows that the biased formulation has a slightly smaller mean gap for $K=4$ before refinement. However, the adjacency formulation requires less solve time. Overall, the adjacency formulation provides the best balance between solution quality and computational cost. We therefore use it as the default objective of MO-QUCO in the main experiments.

\subsection{Generated Conflicting-Objective Instances}
\label{app:conflict}
The QAMOO benchmark used in our main experiments \citep{kotil2025quantum} draws its $K$ weight layers i.i.d.\ from $\mathcal{N}(0,1)$, which makes the objectives weakly conflicting. To explore more conflicting objectives, we follow \citet{tao2026multi} and \citet{de2026quadratic} to generate $K=3$ instances with strongly negatively correlated objectives.

\subsubsection{Instance Construction}
\label{app:conflict-construction}
The novel dataset introduced in this work consists of synthetic MO-MaxCut instances with $K=3$ objectives and $n\in\{20,100,200\}$ vertices. For each value of $n$, we generate complete graphs and Erd\H{o}s--R\'enyi graphs with edge probability $0.5$. We use the five instance seeds $\{0,1,2,3,4\}$ for every size and topology. Each seed determines the graph topology and all edge weights through the deterministic construction presented below. The supplementary code includes the instance generator. Thus, every new instance is uniquely specified by its size, topology, and seed.

For each seed, we first generate either a complete graph or an Erd\H{o}s--R\'enyi graph in which each possible edge is included independently with probability $0.5$. The same edge set is shared by all three objectives. For each existing edge, we sample $c_1$ uniformly from $\{0,\ldots,45\}$. We then sample $c_2$ conditionally as
\[
c_2 \sim
\begin{cases}
\mathrm{U}\{45-c_1,\ldots,45\}, & c_1\leq 22,\\
\mathrm{U}\{0,\ldots,45-c_1\}, & c_1\geq 23.
\end{cases}
\]
Thus, a small $c_1$ tends to produce a large $c_2$ and vice versa. We assign the three edge weights as $w_1=c_1+5$, $w_2=c_2+5$, and $w_3=-0.5\,w_1-5\,w_2$. This construction gives $f_3(z)=-0.5f_1(z)-5f_2(z)$ for every cut $z$. It yields an empirical correlation of $\rho(f_1+f_2,f_3)\approx-0.915$ across the seeded instances. We use five seeds for each configuration. These instances have large Pareto fronts and allow us to assess the outcome coverage achieved by a finite set of sampled scalarizations.

\subsubsection{Exact-Front Validation at \texorpdfstring{$n=20$}{n=20}}
\label{app:conflict-small}
At $n=20$, we enumerate all distinct cuts to obtain the exact Pareto front. Table~\ref{tab:conflict-exact} reports hypervolume as a percentage of the exact-front hypervolume and the Pareto recall. Recall is the fraction of exact Pareto-front outcomes recovered by a method. Results are averaged over five independently generated instances for each topology. Exact WSM and MO-QUCO use the same set of $J=256$ preference vectors. Exact WSM returns one global maximizer for each sampled scalarization. MO-QUCO uses $B=16$ initial points for each preference. The results show that returning one global maximizer per sampled preference recovers a negligible fraction of the exact front. MO-QUCO uses the adjacency formulation and applies five $1$-bit-flip passes. It recovers more than $99.9\%$ of the exact front and attains nearly $100\%$ of its hypervolume.

\begin{table}[!htbp]
\centering
\caption{Exact-front comparison at $n=20$ and $K=3$. For each topology, we independently generate five weighted graph instances and run each method once per instance. Results report mean \%HV and Pareto recall. The $+$ flip column uses five $1$-bit-flip passes.}
\label{tab:conflict-exact}
{\small
\setlength{\tabcolsep}{3pt}
\begin{tabular}{lcccc}
\toprule
 & Random & WSM & MO-QUCO & $+$ flip \\
\midrule
ER $p{=}0.5$, \%HV & $89.96$ & $81.53$ & $97.96$ & $\mathbf{100.0000}$ \\
\quad recall & $0.0167$ & $0.0001$ & $0.0080$ & $\mathbf{1.0000}$ \\
Complete, \%HV & $89.52$ & $79.24$ & $97.82$ & $\mathbf{100.0000}$ \\
\quad recall & $0.0172$ & $0.0001$ & $0.0076$ & $\mathbf{0.9999}$ \\
\bottomrule
\end{tabular}
}
\end{table}

\FloatBarrier
\input{aaai2027/arxiv_appendix_anytime_100_200_10000s.tex}

\section{Limitations}

MO-QUCO does not guarantee recovery of the global Pareto front because it optimizes each scalarized problem approximately. Its performance depends on the numbers of sampled preferences and initial points. Moreover, the $1$-bit-flip refinement can also become expensive for large archives because it evaluates the neighborhood of every vector in the archive.

\end{document}

%% file: aaai2027/arxiv_table_benchmark42.tex
\begin{table}[!htbp]
\centering
\caption{Final performance on the 42-node MO-MaxCut benchmark. %
Entries with uncertainty report the mean $\pm$ sample standard deviation. The $\varepsilon$-CM and MO-QUCO entries use three seeds. The QAOA entries use five repetitions. Entries without uncertainty are single reported values.
A dagger ($\dagger$) marks a run stopped at its sampling budget. A double dagger ($\ddagger$) marks an error within $2 \times 10^{-8}$. A dash (--) marks an unavailable value due to no reported values in the original paper. Among the non-exact methods, the best value in each \%$\mathrm{HV}_{\max}$ and Time column is in bold and the second best is underlined. The pMO-QUCO rows use one NVIDIA H200 GPU for the preference sweep and the CPU for bit-flip refinement.}
\label{tab:main}

\begin{subtable}{\linewidth}
\centering
\caption{$K=3$ with $\mathrm{HV}_{\max}=43{,}471.704$}
\small
\setlength{\tabcolsep}{6pt}
\begin{tabular}{@{}lrrr@{}}
\toprule
Method & $\%\mathrm{HV}_{\max}\uparrow$ & $\#\mathrm{NDP}\uparrow$ & Time (s) $\downarrow$ \\
\midrule
DPA-a (exact) & 100.0000 & 2,063 & 167.5 \\
DCM (exact) & 100.0000 & 2,063 & 368.5 \\
$\varepsilon$-CM & 99.9993 $\pm$ 0.0001 & 2,126 $\pm$ 28 & 10,000 $\pm$ 0\rlap{$^{\dagger}$} \\
WSM & 99.1246 & 154 & 5,005 \\
QAOA (simulation) & $\boldsymbol{100.0000 \pm 0.0000}$\rlap{$^{\ddagger}$} & -- & 360 $\pm$ 105 \\
QAOA (\texttt{ibm\_fez}) & 99.9996 $\pm$ 0.0004 & -- & 2,500 $\pm$ 0\rlap{$^{\dagger}$} \\
\midrule
MO-QUCO (CPU) & $\underline{99.9999 \pm 0.0001}$ & 2,692 $\pm$ 3 & 118.0 $\pm$ 0.3 \\
\;\; + bit-flip & $\boldsymbol{100.0000 \pm 0.0000}$\rlap{$^{\ddagger}$} & 2,591 $\pm$ 1 & 119.8 $\pm$ 0.3 \\
pMO-QUCO (GPU) & $\underline{99.9999 \pm 0.0000}$ & 2,885 $\pm$ 1 & $\boldsymbol{0.9 \pm 0.0}$ \\
\;\; + bit-flip & $\boldsymbol{100.0000 \pm 0.0000}$\rlap{$^{\ddagger}$} & 2,599 $\pm$ 2 & $\underline{2.7 \pm 0.0}$ \\
\bottomrule
\end{tabular}
\end{subtable}
\par\medskip
\begin{subtable}{\linewidth}
\centering
\caption{$K=4$ with $\mathrm{HV}_{\max}=1{,}266{,}143.326$}
\small
\setlength{\tabcolsep}{6pt}
\begin{tabular}{@{}lrrr@{}}
\toprule
Method & $\%\mathrm{HV}_{\max}\uparrow$ & $\#\mathrm{NDP}\uparrow$ & Time (s) $\downarrow$ \\
\midrule
DPA-a (exact) & 100.0000 & 30,288 & 15,579 \\
DCM (exact) & 99.9813 & 17,424 & 10,016\rlap{$^{\dagger}$} \\
$\varepsilon$-CM & 99.9588 $\pm$ 0.0035 & 11,752 $\pm$ 368 & 10,000 $\pm$ 0\rlap{$^{\dagger}$} \\
WSM & 98.6132 & 1,235 & 5,597 \\
QAOA (simulation) & $\boldsymbol{100.0000 \pm 0.0000}$\rlap{$^{\ddagger}$} & -- & 10,000 $\pm$ 0\rlap{$^{\dagger}$} \\
QAOA (\texttt{ibm\_fez}) & -- & -- & -- \\
\midrule
MO-QUCO (CPU) & 99.9954 $\pm$ 0.0003 & 42,095 $\pm$ 56 & $\underline{117.7 \pm 0.6}$ \\
\;\; + bit-flip & $\boldsymbol{100.0000 \pm 0.0000}$\rlap{$^{\ddagger}$} & 47,668 $\pm$ 6 & 271.0 $\pm$ 1.1 \\
pMO-QUCO (GPU) & $\underline{99.9963 \pm 0.0011}$ & 43,904 $\pm$ 40 & $\boldsymbol{0.9 \pm 0.0}$ \\
\;\; + bit-flip & $\boldsymbol{100.0000 \pm 0.0000}$\rlap{$^{\ddagger}$} & 47,658 $\pm$ 12 & 154.2 $\pm$ 0.5 \\
\bottomrule
\end{tabular}
\end{subtable}
\end{table}

%% file: aaai2027/arxiv_appendix_anytime_100_200_10000s.tex
\subsubsection{Anytime Comparison at $n=100$ and $n=200$}
\label{app:anytime-large}

We evaluate the main-paper hyperparameters on the same twenty conflicting-objective instances with a nominal $10{,}000$\,s limit.
Each configuration contains five independently generated graphs with $K=3$.
Classical baselines comprise DPA-a and DCM together with WSM and $\varepsilon$-CM.
The ER curves are shown in Figure~\ref{fig:larger-main} in the main text. Figure~\ref{fig:anytime-large} gives the complementary complete-graph results.

\paragraph{Parameters and stopping rules.}
We use the same MO-QUCO hyperparameters as in the 42-node experiment.
MO-QUCO processes one preference at a time, while pMO-QUCO batches $512$ preferences.

\paragraph{Timing.}
Each CPU method uses one pinned core on an AMD EPYC 9555 processor.
Each GPU run uses one NVIDIA H200 NVL.
Wall-clock timing begins after library and CUDA-context initialization.
It includes method setup, objective evaluation and archive maintenance.
It excludes final serialization and offline HV computation.

\paragraph{Evaluation.}
For instance $i$, let $\mathcal{R}_i$ be the nondominated union of the final feasible archives from all methods.
Let $\mathcal{A}_{im}(t)$ be the archive of method $m$ at time $t$ and let $\mathcal{M}$ contain the displayed methods.
We use
\begin{equation}
\begin{aligned}
r_{ik} &:= \min_{y\in\mathcal{R}_i}y_k-1,\\
\mathrm{HV}_{\max,i} &:= \max_{m\in\mathcal{M}}
\mathrm{HV}_{r_i}\!\left(\mathcal{A}^{\mathrm{final}}_{im}\right),\\
g_{im}(t) &:= \mathrm{HV}_{\max,i}-\mathrm{HV}_{r_i}(\mathcal{A}_{im}(t))+1,\\
q_{im} &:= 100\,
\frac{\mathrm{HV}_{r_i}(\mathcal{A}^{\mathrm{final}}_{im})}{\mathrm{HV}_{\max,i}}.
\end{aligned}
\label{eq:anytime-hv-reference}
\end{equation}
Across the five instances in each configuration, Figures~\ref{fig:larger-main} and~\ref{fig:anytime-large} show the pointwise mean of $g_{im}(t)$ together with its minimum-to-maximum range.
Tables~\ref{tab:anytime-er05} and~\ref{tab:anytime-complete} report the mean and sample standard deviation of $q_{im}$.
Thus, $100\%$ denotes the best observed final HV rather than the unknown full-front hypervolume.

\begin{figure}[!htbp]
\centering
\includegraphics[width=0.94\textwidth]{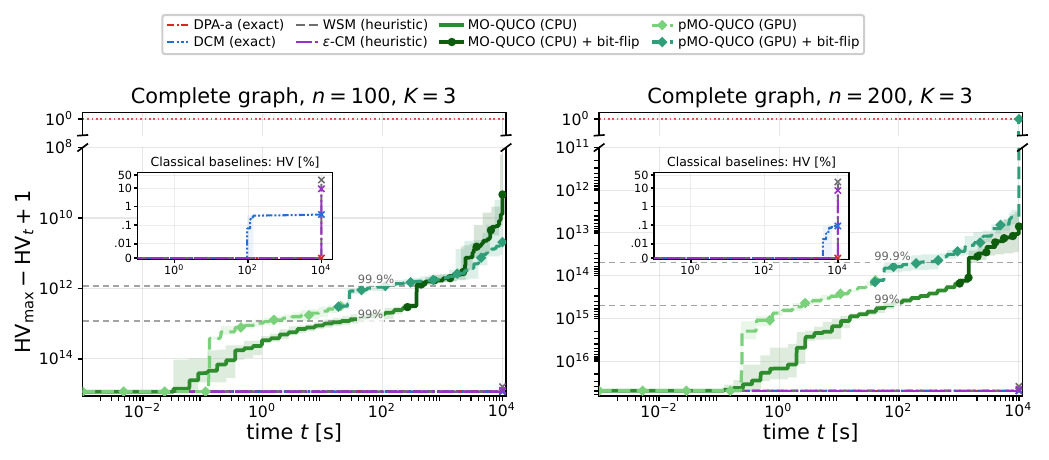}
\caption{Convergence on generated complete graphs with $n=100$ (left) and $n=200$ (right), both with $K=3$.
Both axes are logarithmic and the HV-gap axis is reversed.
Lines and bands show the mean and min-to-max range over five graphs.
The gray dashed lines correspond to $99\%$ and $99.9\%$ of the mean per-instance $\mathrm{HV}_{\max}$ in each panel.
Slashes omit empty gap ranges.
Insets show classical-baseline HV percentages.}
\label{fig:anytime-large}
\end{figure}

\paragraph{Results.}
Tables~\ref{tab:anytime-er05} and~\ref{tab:anytime-complete} show that MO-QUCO and pMO-QUCO achieve over $99.6\%$ of $\mathrm{HV}_{\max}$ on average across all four configurations.
pMO-QUCO requires only $19.2$ to $50.3$\,s compared with $262.2$ to $1{,}042.1$\,s for MO-QUCO.
Adding $1$-bit-flip further increases the mean HV to over $99.98\%$ of $\mathrm{HV}_{\max}$ but uses nearly the full $10{,}000$\,s time limit.
Figures~\ref{fig:larger-main} and~\ref{fig:anytime-large} show rapid initial improvement and further HV gains after adding $1$-bit-flip.

\input{aaai2027/arxiv_table_anytime_er05_10000s.tex}

In contrast, the classical baselines achieve much lower HV under the same time limit.
DPA-a's final archive is empty and DCM achieves at most $1.124\%$ of $\mathrm{HV}_{\max}$ on average.
WSM performs best among the classical baselines but achieves only $17.423\%$ to $27.247\%$ of $\mathrm{HV}_{\max}$ on average.
WSM and $\varepsilon$-CM process only a few scalarized subproblems within the time limit and return very small archives.
Our methods return hundreds of thousands of NDPs, which increase to tens of millions after adding $1$-bit-flip.

\input{aaai2027/arxiv_table_anytime_complete_10000s.tex}

%% file: aaai2027/arxiv_table_anytime_er05_10000s.tex
\begin{table}[!htbp]
\centering
\caption{Final performance on generated Erd\H{o}s--R\'enyi graphs ($p=0.5$) with $n\in\{100,200\}$ and $K=3$ under a nominal $10{,}000$\,s limit. Results report the mean $\pm$ sample standard deviation over five independent graph instances. A dagger ($\dagger$) marks a run stopped at its time budget. Among the non-exact methods, the best mean value in each $\%\mathrm{HV}_{\max}$ and Time column is in bold and the second best is underlined.}
\label{tab:anytime-er05}

\begin{subtable}{\linewidth}
\centering
\caption{$n=100$}
\small
\setlength{\tabcolsep}{6pt}
\begin{tabular}{@{}lrrr@{}}
\toprule
Method & $\%\mathrm{HV}_{\max}\uparrow$ & $\#\mathrm{NDP}\uparrow$ & Time (s) $\downarrow$ \\
\midrule
DPA-a & $0.0000 \pm 0.0000$ & $0 \pm 0$ & $10{,}000.1 \pm 0.0^{\dagger}$ \\
DCM & $1.1240 \pm 0.0854$ & $61 \pm 2$ & $10{,}000.0 \pm 0.0^{\dagger}$ \\
WSM & $23.3103 \pm 12.0793$ & $2 \pm 0$ & $10{,}000.0 \pm 0.0^{\dagger}$ \\
$\varepsilon$-CM & $11.9944 \pm 22.7043$ & $1 \pm 1$ & $10{,}000.0 \pm 0.0^{\dagger}$ \\
\midrule
MO-QUCO (CPU) & $99.6424 \pm 0.0366$ & \ndpentry{272{,}990}{45{,}834} & $\underline{264.1 \pm 1.4}$ \\
\quad $+$ bit-flip & $\underline{99.9977 \pm 0.0035}$ & \ndpentry{28{,}758{,}514}{951{,}050} & $10{,}023.0 \pm 17.8^{\dagger}$ \\
pMO-QUCO (GPU) & $99.6437 \pm 0.0451$ & \ndpentry{272{,}963}{46{,}050} & $\mathbf{22.2 \pm 1.3}$ \\
\quad $+$ bit-flip & $\mathbf{99.9999 \pm 0.0003}$ & \ndpentry{27{,}840{,}498}{1{,}398{,}704} & $10{,}014.5 \pm 13.4^{\dagger}$ \\
\bottomrule
\end{tabular}
\end{subtable}
\par\medskip
\begin{subtable}{\linewidth}
\centering
\caption{$n=200$}
\small
\setlength{\tabcolsep}{6pt}
\begin{tabular}{@{}lrrr@{}}
\toprule
Method & $\%\mathrm{HV}_{\max}\uparrow$ & $\#\mathrm{NDP}\uparrow$ & Time (s) $\downarrow$ \\
\midrule
DPA-a & $0.0000 \pm 0.0000$ & $0 \pm 0$ & $10{,}000.0 \pm 0.0^{\dagger}$ \\
DCM & $0.0745 \pm 0.0025$ & $8 \pm 2$ & $10{,}000.0 \pm 0.0^{\dagger}$ \\
WSM & $17.4230 \pm 9.3925$ & $2 \pm 0$ & $10{,}000.0 \pm 0.0^{\dagger}$ \\
$\varepsilon$-CM & $6.4772 \pm 14.4834$ & $1 \pm 1$ & $10{,}000.0 \pm 0.0^{\dagger}$ \\
\midrule
MO-QUCO (CPU) & $99.6390 \pm 0.0292$ & \ndpentry{263{,}400}{3{,}862} & $\underline{1{,}042.1 \pm 9.7}$ \\
\quad $+$ bit-flip & $\underline{99.9833 \pm 0.0253}$ & \ndpentry{32{,}904{,}444}{1{,}877{,}098} & $10{,}021.9 \pm 12.9^{\dagger}$ \\
pMO-QUCO (GPU) & $99.6427 \pm 0.0596$ & \ndpentry{263{,}333}{3{,}926} & $\mathbf{50.3 \pm 13.4}$ \\
\quad $+$ bit-flip & $\mathbf{99.9888 \pm 0.0109}$ & \ndpentry{32{,}306{,}173}{4{,}699{,}533} & $10{,}025.1 \pm 12.4^{\dagger}$ \\
\bottomrule
\end{tabular}
\end{subtable}
\end{table}

%% file: aaai2027/arxiv_table_anytime_complete_10000s.tex
\begin{table}[!htbp]
\centering
\caption{Final performance on generated complete graphs with $n\in\{100,200\}$ and $K=3$ under a nominal $10{,}000$\,s limit. Results report the mean $\pm$ sample standard deviation over five independent graph instances. A dagger ($\dagger$) marks a run stopped at its time budget. Among the non-exact methods, the best mean value in each $\%\mathrm{HV}_{\max}$ and Time column is in bold and the second best is underlined.}
\label{tab:anytime-complete}

\begin{subtable}{\linewidth}
\centering
\caption{$n=100$}
\small
\setlength{\tabcolsep}{6pt}
\begin{tabular}{@{}lrrr@{}}
\toprule
Method & $\%\mathrm{HV}_{\max}\uparrow$ & $\#\mathrm{NDP}\uparrow$ & Time (s) $\downarrow$ \\
\midrule
DPA-a & $0.0000 \pm 0.0000$ & $0 \pm 0$ & $10{,}000.1 \pm 0.0^{\dagger}$ \\
DCM & $0.3744 \pm 0.0045$ & $26 \pm 1$ & $10{,}000.0 \pm 0.0^{\dagger}$ \\
WSM & $27.2467 \pm 8.9899$ & $2 \pm 0$ & $10{,}000.0 \pm 0.0^{\dagger}$ \\
$\varepsilon$-CM & $9.0631 \pm 14.8933$ & $1 \pm 1$ & $10{,}000.0 \pm 0.0^{\dagger}$ \\
\midrule
MO-QUCO (CPU) & $99.6276 \pm 0.0547$ & \ndpentry{182{,}552}{19{,}172} & $\underline{262.2 \pm 1.2}$ \\
\quad $+$ bit-flip & $\mathbf{99.9997 \pm 0.0006}$ & \ndpentry{28{,}506{,}087}{1{,}105{,}543} & $10{,}017.7 \pm 8.8^{\dagger}$ \\
pMO-QUCO (GPU) & $99.6281 \pm 0.0584$ & \ndpentry{182{,}830}{19{,}268} & $\mathbf{19.2 \pm 0.7}$ \\
\quad $+$ bit-flip & $\underline{99.9943 \pm 0.0060}$ & \ndpentry{21{,}500{,}605}{2{,}825{,}957} & $9{,}990.2 \pm 24.6^{\dagger}$ \\
\bottomrule
\end{tabular}
\end{subtable}
\par\medskip
\begin{subtable}{\linewidth}
\centering
\caption{$n=200$}
\small
\setlength{\tabcolsep}{6pt}
\begin{tabular}{@{}lrrr@{}}
\toprule
Method & $\%\mathrm{HV}_{\max}\uparrow$ & $\#\mathrm{NDP}\uparrow$ & Time (s) $\downarrow$ \\
\midrule
DPA-a & $0.0000 \pm 0.0000$ & $0 \pm 0$ & $10{,}000.0 \pm 0.0^{\dagger}$ \\
DCM & $0.0901 \pm 0.0016$ & $2 \pm 0$ & $10{,}000.0 \pm 0.0^{\dagger}$ \\
WSM & $21.2671 \pm 7.8826$ & $2 \pm 0$ & $10{,}000.0 \pm 0.0^{\dagger}$ \\
$\varepsilon$-CM & $7.4094 \pm 16.5678$ & $1 \pm 1$ & $10{,}000.0 \pm 0.0^{\dagger}$ \\
\midrule
MO-QUCO (CPU) & $99.7341 \pm 0.0296$ & \ndpentry{173{,}883}{1{,}858} & $\underline{1{,}033.8 \pm 3.9}$ \\
\quad $+$ bit-flip & $\underline{99.9858 \pm 0.0184}$ & \ndpentry{31{,}173{,}653}{4{,}648{,}897} & $10{,}010.4 \pm 36.6^{\dagger}$ \\
pMO-QUCO (GPU) & $99.7423 \pm 0.0170$ & \ndpentry{173{,}981}{1{,}808} & $\mathbf{40.0 \pm 0.3}$ \\
\quad $+$ bit-flip & $\mathbf{100.0000 \pm 0.0000}$ & \ndpentry{34{,}752{,}592}{175{,}320} & $10{,}030.5 \pm 6.2^{\dagger}$ \\
\bottomrule
\end{tabular}
\end{subtable}
\end{table}

%% file: main.bbl
\begin{thebibliography}{35}
\providecommand{\natexlab}[1]{#1}
\providecommand{\url}[1]{\texttt{#1}}
\expandafter\ifx\csname urlstyle\endcsname\relax
  \providecommand{\doi}[1]{doi: #1}\else
  \providecommand{\doi}{doi: \begingroup \urlstyle{rm}\Url}\fi

\bibitem[Karp(2009)]{karp2009reducibility}
Richard~M Karp.
\newblock Reducibility among combinatorial problems.
\newblock In \emph{50 Years of Integer Programming 1958-2008: from the Early
  Years to the State-of-the-Art}, pages 219--241. Springer, 2009.

\bibitem[Tarjan and Trojanowski(1977)]{tarjan1977finding}
Robert~Endre Tarjan and Anthony~E Trojanowski.
\newblock Finding a maximum independent set.
\newblock \emph{SIAM Journal on Computing}, 6\penalty0 (3):\penalty0 537--546,
  1977.

\bibitem[Dantzig et~al.(1954)Dantzig, Fulkerson, and
  Johnson]{dantzig1954solution}
George Dantzig, Ray Fulkerson, and Selmer Johnson.
\newblock Solution of a large-scale traveling-salesman problem.
\newblock \emph{Journal of the operations research society of America},
  2\penalty0 (4):\penalty0 393--410, 1954.

\bibitem[Barahona et~al.(1988)Barahona, Gr{\"o}tschel, J{\"u}nger, and
  Reinelt]{barahona1988application}
Francisco Barahona, Martin Gr{\"o}tschel, Michael J{\"u}nger, and Gerhard
  Reinelt.
\newblock An application of combinatorial optimization to statistical physics
  and circuit layout design.
\newblock \emph{Operations Research}, 36\penalty0 (3):\penalty0 493--513, 1988.

\bibitem[Snir and Rao(2006)]{snir2006using}
Sagi Snir and Satish Rao.
\newblock Using max cut to enhance rooted trees consistency.
\newblock \emph{IEEE/ACM transactions on computational biology and
  bioinformatics}, 3\penalty0 (4):\penalty0 323--333, 2006.

\bibitem[Ehrgott(2005)]{ehrgott2005multicriteria}
Matthias Ehrgott.
\newblock \emph{Multicriteria optimization}.
\newblock Springer, 2005.

\bibitem[Jozefowiez et~al.(2008)Jozefowiez, Semet, and
  Talbi]{jozefowiez2008multi}
Nicolas Jozefowiez, Fr{\'e}d{\'e}ric Semet, and El-Ghazali Talbi.
\newblock Multi-objective vehicle routing problems.
\newblock \emph{European journal of operational research}, 189\penalty0
  (2):\penalty0 293--309, 2008.

\bibitem[Du et~al.(2014)Du, Li, Xiu, and Xu]{du2014simultaneous}
Dong-Lei Du, Yu~Li, Nai-Hua Xiu, and Da-Chuan Xu.
\newblock Simultaneous approximation of multi-criteria submodular function
  maximization.
\newblock \emph{Journal of the Operations Research Society of China},
  2\penalty0 (3):\penalty0 271--290, 2014.

\bibitem[Censor-Hillel et~al.(2021)Censor-Hillel, Marelly, Schwartz, and
  Tonoyan]{censorhillel2021fault}
Keren Censor-Hillel, Noa Marelly, Roy Schwartz, and Tigran Tonoyan.
\newblock Fault tolerant max-cut.
\newblock \emph{arXiv preprint arXiv:2105.01138}, 2021.

\bibitem[Angel et~al.(2006)Angel, Bampis, and
  Gourv{\`e}s]{angel2006approximation}
Eric Angel, Evripidis Bampis, and Laurent Gourv{\`e}s.
\newblock Approximation algorithms for the bi-criteria weighted max-cut
  problem.
\newblock \emph{Discrete Applied Mathematics}, 154\penalty0 (12):\penalty0
  1685--1692, 2006.

\bibitem[Kotil et~al.(2025)Kotil, Pelofske, Riedm{\"u}ller, Egger, Eidenbenz,
  Koch, and Woerner]{kotil2025quantum}
Ayse Kotil, Elijah Pelofske, Stephanie Riedm{\"u}ller, Daniel~J Egger, Stephan
  Eidenbenz, Thorsten Koch, and Stefan Woerner.
\newblock Quantum approximate multi-objective optimization.
\newblock \emph{Nature Computational Science}, pages 1--10, 2025.

\bibitem[Farhi et~al.(2014)Farhi, Goldstone, and Gutmann]{farhi2014quantum}
Edward Farhi, Jeffrey Goldstone, and Sam Gutmann.
\newblock A quantum approximate optimization algorithm.
\newblock \emph{arXiv preprint arXiv:1411.4028}, 2014.

\bibitem[Goemans and Williamson(1995)]{goemans1995improved}
Michel~X Goemans and David~P Williamson.
\newblock Improved approximation algorithms for maximum cut and satisfiability
  problems using semidefinite programming.
\newblock \emph{Journal of the ACM (JACM)}, 42\penalty0 (6):\penalty0
  1115--1145, 1995.

\bibitem[Alkhouri et~al.(2025)Alkhouri, Wu, Yu, Liu, Wang, and
  Velasquez]{alkhouri2025scalable}
Ismail Alkhouri, Mian Wu, Cunxi Yu, Jia Liu, Rongrong Wang, and Alvaro
  Velasquez.
\newblock A scalable lift-and-project differentiable approach for the maximum
  cut problem.
\newblock \emph{arXiv preprint arXiv:2509.18612}, 2025.

\bibitem[Allmendinger et~al.(2022)Allmendinger, Jaszkiewicz, Liefooghe, and
  Tammer]{allmendinger2022if}
Richard Allmendinger, Andrzej Jaszkiewicz, Arnaud Liefooghe, and Christiane
  Tammer.
\newblock What if we increase the number of objectives? theoretical and
  empirical implications for many-objective combinatorial optimization.
\newblock \emph{Computers \& Operations Research}, 145:\penalty0 105857, 2022.

\bibitem[Lin et~al.(2022)Lin, Yang, and Zhang]{lin2022pareto}
Xi~Lin, Zhiyuan Yang, and Qingfu Zhang.
\newblock Pareto set learning for neural multi-objective combinatorial
  optimization.
\newblock \emph{arXiv preprint arXiv:2203.15386}, 2022.

\bibitem[Singh et~al.(2026)Singh, Wu, Yang, Rosing, Yu, and
  Ma]{singh2026divide}
Esha Singh, Dongxia Wu, Chien-Yi Yang, Tajana Rosing, Rose Yu, and Yi-An Ma.
\newblock Divide and learn: Multi-objective combinatorial optimization at
  scale.
\newblock \emph{arXiv preprint arXiv:2602.11346}, 2026.

\bibitem[D{\"a}chert et~al.(2024)D{\"a}chert, Fleuren, and
  Klamroth]{dachert2024simple}
Kerstin D{\"a}chert, Tino Fleuren, and Kathrin Klamroth.
\newblock A simple, efficient and versatile objective space algorithm for
  multiobjective integer programming.
\newblock \emph{Mathematical Methods of Operations Research}, 100\penalty0
  (1):\penalty0 351--384, 2024.

\bibitem[Boland et~al.(2017)Boland, Charkhgard, and Savelsbergh]{boland2017new}
Natashia Boland, Hadi Charkhgard, and Martin Savelsbergh.
\newblock A new method for optimizing a linear function over the efficient set
  of a multiobjective integer program.
\newblock \emph{European journal of operational research}, 260\penalty0
  (3):\penalty0 904--919, 2017.

\bibitem[K{\"o}nen and Stiglmayr(2025)]{konen2025supportedness}
David K{\"o}nen and Michael Stiglmayr.
\newblock On supportedness in multi-objective combinatorial optimization.
\newblock \emph{arXiv preprint arXiv:2501.13842}, 2025.

\bibitem[Schuetz et~al.(2022)Schuetz, Brubaker, and
  Katzgraber]{schuetz2022combinatorial}
Martin~JA Schuetz, J~Kyle Brubaker, and Helmut~G Katzgraber.
\newblock Combinatorial optimization with physics-inspired graph neural
  networks.
\newblock \emph{Nature Machine Intelligence}, 4\penalty0 (4):\penalty0
  367--377, 2022.

\bibitem[Alkhouri et~al.(2024)Alkhouri, Denmat, Li, Yu, Liu, Wang, and
  Velasquez]{alkhouri2024differentiable}
Ismail Alkhouri, Cedric~Le Denmat, Yingjie Li, Cunxi Yu, Jia Liu, Rongrong
  Wang, and Alvaro Velasquez.
\newblock Differentiable quadratic optimization for the maximum independent set
  problem.
\newblock \emph{arXiv preprint arXiv:2406.19532}, 2024.

\bibitem[Sun et~al.(2026)Sun, Alkhouri, Huang, Velasquez, Jha, and
  Wang]{sun2026mutation}
Yongliang Sun, Ismail Alkhouri, Cheng-Han Huang, Alvaro Velasquez, Susmit Jha,
  and Rongrong Wang.
\newblock Mutation-guided differentiable quadratic combinatorial optimization.
\newblock \emph{arXiv preprint arXiv:2605.06921}, 2026.

\bibitem[Zadeh(1963)]{zadeh1963optimality}
Lofti Zadeh.
\newblock Optimality and non-scalar-valued performance criteria.
\newblock \emph{IEEE transactions on Automatic Control}, 8\penalty0
  (1):\penalty0 59--60, 1963.

\bibitem[Lin et~al.(2019)Lin, Zhen, Li, Zhang, and Kwong]{lin2019pareto}
Xi~Lin, Hui-Ling Zhen, Zhenhua Li, Qing-Fu Zhang, and Sam Kwong.
\newblock Pareto multi-task learning.
\newblock \emph{Advances in neural information processing systems}, 32, 2019.

\bibitem[Navon et~al.(2020)Navon, Shamsian, Fetaya, and
  Chechik]{navon2020learning}
Aviv Navon, Aviv Shamsian, Ethan Fetaya, and Gal Chechik.
\newblock Learning the pareto front with hypernetworks.
\newblock In \emph{International conference on learning representations}, 2020.

\bibitem[Zhang and Li(2007)]{zhang2007moea}
Qingfu Zhang and Hui Li.
\newblock Moea/d: A multiobjective evolutionary algorithm based on
  decomposition.
\newblock \emph{IEEE Transactions on evolutionary computation}, 11\penalty0
  (6):\penalty0 712--731, 2007.

\bibitem[Manual(1987)]{manual1987ibm}
CPLEX~User’s Manual.
\newblock Ibm ilog cplex optimization studio.
\newblock \emph{Version}, 12\penalty0 (1987-2018):\penalty0 1, 1987.

\bibitem[Haimes(1971)]{haimes1971bicriterion}
Yacov Haimes.
\newblock On a bicriterion formulation of the problems of integrated system
  identification and system optimization.
\newblock \emph{IEEE transactions on systems, man, and cybernetics}, \penalty0
  (3):\penalty0 296--297, 1971.

\bibitem[Gurobi~Optimization(2023)]{gurobi2023gurobi}
LLC Gurobi~Optimization.
\newblock Gurobi optimizer reference manual.
\newblock 2023.

\bibitem[Zitzler and Thiele(1999)]{zitzler1999multiobjective}
Eckart Zitzler and Lothar Thiele.
\newblock Multiobjective evolutionary algorithms: a comparative case study and
  the strength pareto approach.
\newblock \emph{IEEE transactions on Evolutionary Computation}, 3\penalty0
  (4):\penalty0 257--271, 1999.

\bibitem[Yang et~al.(2019)Yang, Emmerich, Deutz, and
  B{\"a}ck]{yang2019efficient}
Kaifeng Yang, Michael Emmerich, Andr{\'e} Deutz, and Thomas B{\"a}ck.
\newblock Efficient computation of expected hypervolume improvement using box
  decomposition algorithms.
\newblock \emph{Journal of Global Optimization}, 75\penalty0 (1):\penalty0
  3--34, 2019.

\bibitem[Biscani and Izzo(2020)]{biscani2020parallel}
Francesco Biscani and Dario Izzo.
\newblock A parallel global multiobjective framework for optimization: pagmo.
\newblock \emph{Journal of Open Source Software}, 5\penalty0 (53):\penalty0
  2338, 2020.

\bibitem[Tao et~al.(2026)Tao, Mosharev, and Yung]{tao2026multi}
Xian-Zhe Tao, Pavel Mosharev, and Man-Hong Yung.
\newblock Multi-objective optimization by quantum-annealing-inspired
  algorithms.
\newblock \emph{arXiv preprint arXiv:2604.26477}, 2026.

\bibitem[De~Santis et~al.(2026)De~Santis, L{\'e}tocart, and
  Zhang]{de2026quadratic}
Marianna De~Santis, Lucas L{\'e}tocart, and Yue Zhang.
\newblock Quadratic convex reformulations for multiobjective binary quadratic
  programming: M. de santis et al.
\newblock \emph{Journal of Global Optimization}, pages 1--32, 2026.

\end{thebibliography}
